%% file: main.tex
\documentclass[conference,compsoc]{IEEEtran}

\usepackage[nocompress]{cite}

\usepackage{amsmath,amssymb}
\usepackage{amsthm}
\usepackage{algorithm}
\usepackage{algorithmic}
\usepackage{url}
\usepackage[hidelinks]{hyperref}

\newtheorem{theorem}{Theorem}
\newtheorem{lemma}{Lemma}

\newenvironment{inlinealgorithm}[1]{%
  \par\medskip
  \refstepcounter{algorithm}%
  \noindent\begin{minipage}{\columnwidth}
  \hrule\smallskip
  \textbf{Algorithm \thealgorithm}\ #1\par
  \smallskip\hrule\smallskip
}{%
  \smallskip\hrule
  \end{minipage}
  \par\medskip
}

\renewcommand{\paragraph}[1]{%
  \par\addvspace{0ex plus 0.1ex minus 0.1ex}%
  \noindent\hspace*{2\parindent}{\normalfont\normalsize\bfseries #1}\ \ignorespaces}

\begin{document}

\title{Near-Tight Theoretical Bounds for Incentive Compatibility in Bitcoin Mining}

\author{\IEEEauthorblockN{Akira Sakurai}
\IEEEauthorblockA{Kyoto University\\
Email: akilucky0304@gmail.com}
\and
\IEEEauthorblockN{Taishi Nakai}
\IEEEauthorblockA{Kyoto University\\
Email: distributed.nakai@gmail.com}
\and
\IEEEauthorblockN{Kazuyuki Shudo}
\IEEEauthorblockA{Kyoto University\\
Email: shudo@computer.org}}

\maketitle

\input{sections/abstract}

\input{sections/introduction}
\input{sections/model}
\input{sections/ideal-analysis}
\input{sections/ideal-to-real}
\input{sections/algorithm}
\input{sections/closed-form-upper-bound}
\input{sections/gamma-monotonicity}
\input{sections/arbitrary-tie-parameters}
\input{sections/near-tightness}
\input{sections/discussion}
\input{sections/related-work}
\input{sections/conclusion}

\bibliographystyle{IEEEtran}
\bibliography{refs}

\end{document}

%% file: sections/abstract.tex
\begin{abstract}
When is honest Bitcoin mining rational? This question is central to the incentive design of proof-of-work blockchains. Sapirshtein et al. computationally derived near-tight lower and upper bounds on the incentive-compatibility threshold using a Markov Decision Process. Kiayias et al.'s Blockchain Mining Games instead derived theoretical lower and upper bounds. However, this theoretical approach has two limitations: its model restricts miners to a narrow action space and assumes idealized tie behavior, and its lower and upper bounds are far from tight.

We resolve both limitations. We develop a more realistic model with a broader miner action space and asymmetric tie-breaking parameters $\gamma^-$ and $\gamma^+$. We then propose an algorithm that computes lower and upper bounds on the incentive-compatibility threshold with a maximum error of $9.98006\times10^{-4}$.
\end{abstract}

%% file: sections/introduction.tex
\section{Introduction}
When is honest Bitcoin mining rational? This is one of the most fundamental questions in blockchain protocol design. Bitcoin~\cite{nakamoto2008bitcoin} is permissionless and has no central administrator, which gives it strong censorship resistance. At the same time, because the system is maintained by self-interested participants, it is crucial to design incentives that align individual behavior with the interests of the system as a whole.

Sapirshtein et al. used a Markov Decision Process analysis to derive
near-tight lower and upper bounds on the incentive-compatibility
threshold~\cite{sapirshtein2016optimal}. Their reported bounds rely on a
numerical MDP solver satisfying a prescribed accuracy guarantee.

Kiayias et al.'s Blockchain Mining Games takes a complementary theoretical
approach and derives provable lower and upper bounds on the
threshold~\cite{kiayias2016blockchain}. However, two challenges remain:
\begin{itemize}
  \item \textbf{Model realism.} Their bounds are established in an idealized
  model that restricts miners to a simplified action space and assumes that,
  when a blockchain tie occurs, all honest miners mine on the honest branch.
  \item \textbf{Tightness.} Their theoretical lower and upper bounds for the
  incentive-compatibility threshold are $0.308$ and $0.455$, respectively, leaving
  an interval of width $0.147$.
\end{itemize}

We make two contributions.

\begin{itemize}
  \item \textbf{REAL mining model.} We allow miners to take broader
  actions and introduce two asymmetric tie-breaking parameters,
  $\gamma^-$ and $\gamma^+$. Here, $\gamma^-$ denotes the fraction of honest
  miners who mine on the deviating miner's branch when the deviating miner
  catches up,
  while $\gamma^+$ denotes the corresponding fraction when the honest chain
  catches up with the deviating miner.
  \item \textbf{Theoretical threshold characterization.} We give
  independent algorithms for computing a lower bound and an upper bound,
  respectively, on the incentive-compatibility threshold. Exact rational evaluation
  of both algorithms over all 1,000,000 pairs in
  $\{0,0.001,\ldots,0.999\}^2$ gives a maximum pointwise gap of
  $2.415\times10^{-7}$. We also prove that the threshold is coordinatewise
  nonincreasing in $(\gamma^-,\gamma^+)$. Coordinatewise monotonicity then
  gives a certified interval of width at most $9.98006\times10^{-4}$ for every
  parameter pair in $[0,1]^2$.
\end{itemize}

%% file: sections/model.tex
\section{Model and Our Goal}
\label{sec:model}

\subsection{REAL Model}
\label{sec:model-real}

The REAL model is an $n$-miner game with strategic core components and
Frontier components.

\subsubsection{Miners}

Let
\[
  \mathcal M = \{m_1,\ldots,m_n\}
\]
be the fixed set of miners. Each miner $m_i\in\mathcal M$ has total
economic hash-power share
\[
  \alpha_i>0,
  \qquad
  \sum_{i=1}^n \alpha_i=1.
\]
A miner is the unit of reward ownership and utility accounting. Its
hash power is split between the core component and the Frontier
component defined below.

\subsubsection{Blocks}

The unique block that references no parent block is called the genesis block.
Every other block references exactly one parent block.

Each block carries a source-component label and a reward-owner label.
The source-component label records the operational source of the block:
it is $c_i$ for a block generated by core component $c_i$, and
Frontier for a block generated by a Frontier component. The reward-owner
label records the miner $m_i$ that owns the generating component and to
whom the block reward is credited.

The source-component label is visible only within the component that
generated the block. The reward-owner label is visible only within the
miner that owns the block.

\subsubsection{Core Component}

The REAL model has exactly two operational components for each miner:
a core component and a Frontier component. The core component
is the only strategic component. It maintains a core-local state and
follows a core policy; its operational hash share is assigned by the
miner's REAL strategy, defined below.

Core states are managed independently. A core component does not obtain
another component's private tree, unreleased blocks, or mining targets,
even when the two components have the same reward owner.

\paragraph{\textbf{Core state.}}
The state of core component $c_i$ is the triple
\[
  x=(T^{\mathrm{pub}},T_{c_i},\sigma_{c_i})\in\mathcal X,
\]
where $T_{c_i}$ is the core-local tree. It contains the public tree and all
unreleased blocks known privately to $c_i$, so
\[
  T^{\mathrm{pub}}\subseteq T_{c_i}.
\]
The public block tree $T^{\mathrm{pub}}$ is global.

The tie-origin coordinate satisfies $\sigma_{c_i}\in\{\bot,-,+\}$.
Initially, $\sigma_{c_i}=\bot$, meaning that no origin has yet been set.
Whenever the maximum public height increases, $c_i$ compares it with the
height of its highest private block, taking the genesis height if no such
block exists. It sets $\sigma_{c_i}=+$ if the private height is at least the
new public height, and $\sigma_{c_i}=-$ otherwise. A private-block discovery
alone does not change the label.

\paragraph{\textbf{Actions.}}
At any state $x$, a core component selects exactly one action. An action
is either a release action or a mining action. Let
$A_{c_i}^{\mathrm{rel}}(x)$ and
$A_{c_i}^{\mathrm{mine}}(x)$ denote the corresponding feasible action
sets, and define
\[
  A_{c_i}(x)
  =
  A_{c_i}^{\mathrm{rel}}(x)
  \,\dot\cup\,
  A_{c_i}^{\mathrm{mine}}(x).
\]

A release action $\rho\in A_{c_i}^{\mathrm{rel}}(x)$ specifies a
finite, possibly empty, set of currently unreleased blocks in $T_{c_i}$ to be
published. The published blocks are added to $T^{\mathrm{pub}}$ in the
state of every component. The release must be ancestry-closed: if a
block is released, then all of its unreleased ancestors are released by
the same action. The action that publishes the empty set is denoted
$\mathsf{NoRelease}$; it leaves the state unchanged but is recorded in the
history.

A release action is applied instantaneously and no proof-of-work
discovery occurs while it is applied. Release actions selected at the
same state are applied simultaneously,
after which policies are evaluated again. This continues until every
core component selects a mining action. Between a state-changing mining
event and the next mining action, each core component may select
$\mathsf{NoRelease}$ only finitely many times. Since each nonempty release
action publishes at least one previously unreleased block and no block is
generated during this phase, this restriction ensures that the release
phase terminates after finitely many steps.

A mining action of core component $c_i$ at state $x$ is either a
target-mining action or the action $\mathsf{NoMining}$. A target-mining
action selects a single target $z\in F_{c_i}(x)$, where $F_{c_i}(x)$ is
the feasible mining-target set defined by the mining-anchor restriction
below. A target-mining action remains in force until the core state
changes. The action $\mathsf{NoMining}$ has no mining target. If the next
mining opportunity belongs to $c_i$, no block is generated, the state is
unchanged, and the core policy selects another action. Otherwise,
$\mathsf{NoMining}$ remains in force until another component discovers a
block and changes the state.

When a component $c$ discovers a block, the block is appended to its
mining target and is initially added as an unreleased block to its
component-local tree $T_c$. The discovered block updates the state, and
policies are then evaluated at the successor state before the next mining
event.

\paragraph{\textbf{Mining anchor.}}
For a state $x$, let $\mathcal P_{\max}(x)$ be the set of public
branches of maximum length. Define
\[
  P_{\mathrm{pub}}(x)
  =
  \operatorname{lcp}\{P : P\in\mathcal P_{\max}(x)\}.
\]
Thus, if the public longest branch is unique, $P_{\mathrm{pub}}(x)$ is
that branch; if a public tie is present, $P_{\mathrm{pub}}(x)$ is the
common branch shared by the tied public longest branches.

For a core component $c_i$, let $V_{c_i}(x)$ be the set of blocks in
$T_{c_i}$ whose source component is $c_i$. For each
$v\in V_{c_i}(x)$, let $\operatorname{path}(v)$ be the unique path from
the genesis block to $v$. If $V_{c_i}(x)$ is nonempty, the mining anchor is
\[
  a_{c_i}(x)
  =
  \operatorname{tip}\!
  \Biggl(
    \max_{\preceq}\!
    \Bigl\{
      \operatorname{lcp}\!
      \bigl(P_{\mathrm{pub}}(x),\operatorname{path}(v)\bigr)
      :v\in V_{c_i}(x)
    \Bigr\}
  \Biggr).
\]
If $V_{c_i}(x)$ is empty, $c_i$ has not yet generated a block and
$a_{c_i}(x)$ is defined as the genesis block.

The feasible mining-target set of core component $c_i$ is
\[
  F_{c_i}(x)
  =
  \{z\in T_{c_i}:a_{c_i}(x)\preceq z\},
\]
where $u\preceq v$ means that $u$ is an ancestor of $v$.
The mining-anchor restriction applies only to mining actions of
core components. It does not restrict release actions or Frontier
components.

\paragraph{\textbf{History.}}
For a root state $x_\star\in\mathcal X$, a finite core history is a
sequence
\[
  h=(x_0,u_0,x_1,\ldots,u_{t-1},x_t)
\]
with $x_0=x_\star$, $u_s\in A_{c_i}(x_s)$, and $x_{s+1}$ a possible
successor of $x_s$ under $u_s$ for every $s<t$. Let
$\mathcal H_{c_i}(x_\star)$ be the set of these histories and write
$\operatorname{cur}(h)=x_t$.

\paragraph{\textbf{Policy.}}
A core policy from $x_\star$ assigns each history a feasible action:
\[
\begin{gathered}
  \pi_i^c:\mathcal H_{c_i}(x_\star)
  \longrightarrow\displaystyle\bigcup_{x\in\mathcal X}A_{c_i}(x),
  \\
  \pi_i^c(h)\in A_{c_i}(\operatorname{cur}(h))
  \quad\text{for every }h\in\mathcal H_{c_i}(x_\star).
\end{gathered}
\]
Two histories ending in the same state may prescribe different actions.

\subsubsection{Frontier Component}

A Frontier component is an honest component that executes the protocol using
the hash power assigned to it by its owner's strategy. It is operationally
separate from the owner's core component. The two components differ in their
policies,
the application of the mining-anchor restriction, and their
tie-breaking behavior.

\paragraph{\textbf{Policy.}}
Unlike a core component, a Frontier component follows the Frontier
policy. If it holds any unreleased private block, it selects the release
action that publishes all such blocks. Otherwise, it selects a mining
action on a public branch of maximum length. If the public longest branch
is unique, the action targets its tip; if a public tie is present, the
action follows the tie-breaking rule below. A block discovered by a
Frontier component has source-component label Frontier and reward-owner
label equal to the owner of that component.

\paragraph{\textbf{Mining anchor.}}
Unlike a core component, a Frontier component is not subject to the
mining-anchor restriction.

\paragraph{\textbf{Tie-breaking.}}
The REAL model has two tie-breaking parameters,
\[
  \gamma^-,\gamma^+ \in [0,1].
\]
For $\sigma\in\{-,+\}$, write
\[
  \gamma^\sigma
  =
  \begin{cases}
    \gamma^-, & \sigma=-,\\
    \gamma^+, & \sigma=+.
  \end{cases}
\]

When a public tie occurs, the core's recorded origin
$\sigma\in\{-,+\}$ determines the Frontier allocation. A fraction
$\gamma^\sigma$ of Frontier hash power mines on the core-side branches,
while the remaining fraction $1-\gamma^\sigma$ mines on the Frontier-side
branches. If more than two branches participate in the tie, the total
Frontier hash power assigned to the core-side and Frontier-side branches
remains $\gamma^\sigma$ and $1-\gamma^\sigma$, respectively. Among multiple
core-side branches, the core-side share mines on the branch whose tied-height
block was generated first if the branches are published simultaneously;
otherwise, it mines on the branch published first. Additional tied branches
do not change either total.

\subsubsection{REAL Strategies and the Frontier Strategy}

A REAL strategy of miner $m_i$ specifies both the core policy $\pi_i^c$
executed by $c_i$ and how $m_i$ splits its economic hash power between
its core component $c_i$ and Frontier component $f_i$. Formally, a REAL
strategy is
\[
  \Sigma_i=(p_i,\pi_i^c),
  \qquad
  0\le p_i\le1.
\]
Here $p_i$ is the fraction of miner $m_i$'s hash power assigned to the
core component. Thus the core component $c_i$ has operational hash share
$\alpha_i p_i$ and executes the core policy $\pi_i^c$. The remaining
operational hash share $\alpha_i(1-p_i)$ is assigned to the
Frontier component $f_i$, which executes the Frontier policy.

The core component and Frontier component of a miner are operationally
separate and do not share state. Their only shared quantity is the final
economic block reward of miner $m_i$.

The core-REAL model is the restricted model in which
$p_i\in\{0,1\}$ for every miner $m_i$.

The REAL Frontier strategy of $m_i$, denoted $\Sigma_i^{\mathsf F}$, is
represented by assigning all hash share to the Frontier
component, i.e., $p_i=0$.

A miner using $\Sigma_i^{\mathsf F}$ is called a Frontier miner. When
no ambiguity arises, we write Frontier for this REAL strategy.

\subsubsection{Execution}

The initial state consists only of the genesis block. The genesis block
belongs to no miner and is not counted as any miner's block.

An execution is a realized sequence of mining-event outcomes defined
independently of the REAL strategy profile. It records the component
selected at each event. A selected component using a target-mining action
creates a block; a selected core using $\mathsf{NoMining}$ instead makes
another policy decision without changing the state.

When a public tie is present and the next block is generated by the
Frontier aggregate, the execution also supplies the random draw used by the
tie-breaking rule. The current tie and this draw determine which tied branch
the Frontier block extends. Thus the strategy profile together with the
execution determines the realized block tree and the utilities.

\subsubsection{Utility}
\label{subsec:real-utility}

The common trunk is the path from the root on which the network has
reached agreement. A path is agreed if, from that point onward, the
actions prescribed by every miner's strategy leave it unchanged under
every possible continuation of the execution.

For transition $t$, let $R_{i,t}$ denote the reward credited to miner
$m_i$ for blocks newly included in the common trunk by that transition,
as determined by reward-owner labels, and let $L_t$ denote the
public-frontier advancement. After $n$ transitions, define
\[
  \mathcal R_{i,n}:=\sum_{t=0}^{n-1}R_{i,t},
  \qquad
  \mathcal L_n:=\sum_{t=0}^{n-1}L_t.
\]
The long-run utility of miner $m_i$ is
\[
  U_i
  =
  \liminf_{n\to\infty}
  \frac{\mathbb E[\mathcal R_{i,n}]}
       {\mathbb E[\mathcal L_n]}.
\]
Thus the denominator measures cumulative public-frontier advancement, not
advancement of the common trunk or the mining anchor.

Regardless of whether a split is used, rewards and utility are determined
by the reward-owner labels.

\subsection{IDEAL Model}
\label{sec:model-ideal}

The IDEAL model restricts the tested core of hash share $p$ to the canonical
actions defined below, while aggregating all remaining hash power into
Frontier. It retains the REAL state, tie-breaking rule, and utility accounting.
We use it both for the certificate LP and as the target of the canonicalization
argument in Section~\ref{sec:ideal-to-real}.

\subsubsection{Canonical Actions}

The IDEAL model restricts the deviating core component to the following
canonical action families:
\[
\begin{gathered}
  \mathsf{Wait},\quad
  \mathsf{Hedge},\quad
  \mathsf{NoMining},\quad
  \mathsf{TieWait},\\
  \mathsf{NoRelease},\quad
  \mathsf{Override},\quad
  \mathsf{Match},\quad
  \mathsf{TieOverride}.
\end{gathered}
\]
No other action or private-forest operation is available in the IDEAL model.

\paragraph{Wait.}
At a non-public-tie state, $\mathsf{Wait}$ is the canonical mining
action that mines on the deviating core's highest block outside the public
longest branch. 

\paragraph{Hedge.}
At a non-public-tie state, $\mathsf{Hedge}$ is the canonical mining action
that mines on a public block such that, upon success, the newly generated
block has greater height than the block generated by a successful
$\mathsf{Wait}$ action from the same state. If no $\mathsf{Wait}$ target
exists, $\mathsf{Hedge}$ may mine on any public block. 

\paragraph{NoMining.}
At any state, $\mathsf{NoMining}$ is the canonical mining action with no
mining target. With probability $p$, the state remains unchanged and the
deviating core selects its next action. With probability $q=1-p$, the core
performs no mining and waits until the state changes.

\paragraph{NoRelease.}
At any state, $\mathsf{NoRelease}$ is the canonical release action that
publishes no blocks. It leaves the state unchanged, after which the core
selects its next action. As in the REAL model, only finitely many consecutive
selections of $\mathsf{NoRelease}$ are allowed between a state-changing
mining event and the next mining action.

\paragraph{Override.}
At a non-public-tie state, $\mathsf{Override}$ is the release action that
publishes exactly enough blocks to make the core-side branch strictly longer
than the Frontier-side branch, which is the public longest branch. All other
private blocks remain unpublished.

\paragraph{Match.}
At a non-public-tie state, $\mathsf{Match}$ is the release action that
publishes exactly enough private blocks to make the core-side branch as high
as the Frontier-side branch. All other private blocks remain unpublished.
This creates a public tie. If the core has caught up
to the Frontier-side branch, the tie has no hidden reserve. If the core is
strictly ahead, the released prefix creates a public tie and the unreleased
part remains as hidden reserve behind the core-side tied branch. Because
Match does not increase the maximum public height, it does not update the
tie-origin coordinate. The recorded value $\sigma_{c_i}\in\{-,+\}$ is
carried into the public tie and determines the parameter $\gamma^{\sigma_{c_i}}$.

\paragraph{TieWait.}
At a public-tie state, $\mathsf{TieWait}$ is the canonical mining action
that mines on the highest block generated by the deviating core. If the
deviating core discovers the next block, the new block is withheld and
the reserve length increases by one. If the Frontier aggregate discovers
the next block, the public tie is resolved according to the tie-breaking
rule: for tie origin $\sigma\in\{-,+\}$, a $\gamma^\sigma$ fraction of
Frontier hash power mines on the core-side tied branch, and the remaining
$1-\gamma^\sigma$ fraction mines on the Frontier-side branch.

\paragraph{TieOverride.}
At a public-tie state, $\mathsf{TieOverride}$ is the release action that
publishes exactly enough private blocks to make the core-side branch strictly
longer than the Frontier-side branch. All other private blocks remain
unpublished.

\subsection{Our Goal}
\label{subsec:our-goal}

For fixed $\gamma^-,\gamma^+\in[0,1]$, our goal is to identify the
incentive-compatibility threshold of the REAL model. For a hash-share bound
$c\in[0,1)$, the REAL model is \emph{incentive-compatible at $c$} if the
Frontier profile is a Nash equilibrium in every REAL instance satisfying
$\alpha_i\le c$ for every miner. Let
\[
  \mathcal I^{\mathrm{REAL}}(\gamma^-,\gamma^+)
\]
denote the set of incentive-compatible hash-share bounds.

The set $\mathcal I^{\mathrm{REAL}}(\gamma^-,\gamma^+)$ contains $0$ and is
prefix-closed. We define the REAL incentive-compatibility threshold by
\[
  \alpha^*(\gamma^-,\gamma^+)
  :=
  \sup\mathcal I^{\mathrm{REAL}}(\gamma^-,\gamma^+)
\]
The REAL model is incentive-compatible at every
$c<\alpha^*$ and is not incentive-compatible at every $c>\alpha^*$; the
endpoint may or may not be incentive-compatible.

We provide algorithms for computing provable lower and upper bounds on
$\alpha^*$ and show that these bounds are near-tight.

%% file: sections/ideal-analysis.tex
\section{Certificate LP in the IDEAL Model}
\label{sec:ideal-analysis}

We construct a certificate LP whose feasibility certifies incentive
compatibility in the IDEAL model.

To prove incentive compatibility in the IDEAL model, we use a
one-deviation argument. Fix an arbitrary miner $m_i$, fix all other
miners to the Frontier strategy, and analyze the best unilateral deviation
of $m_i$ in the IDEAL model.

Throughout this section, all miners in $\mathcal M\setminus\{m_i\}$ follow the Frontier strategy and are aggregated into a single Frontier aggregate. Let
\[
p
\]
be miner $m_i$'s hash share in the IDEAL model, and let
\[
q=1-p.
\]
We assume $p<1/2$.

\subsection{Reduction to the IDEAL One-Deviation Game}
\label{subsec:one-deviation-reduction}

Every reachable IDEAL one-deviation state is represented in the following
canonical state space:
\[
\begin{aligned}
\widehat{\mathcal X}
={}&
\{(0,0)\}
\cup\{(a,0):a\ge1\}\\
&\cup\{(a,b):0\le a<b\}\\
&\cup\{(a,b)^\sigma:a>b\ge1,\ \sigma\in\{-,+\}\}\\
&\cup\{D_b^\sigma:b\ge1,\ \sigma\in\{-,+\}\}\\
&\cup\{E_{b,r}^\sigma:b\ge1,\ r\ge0,\ \sigma\in\{-,+\}\}.
\end{aligned}
\]
The first coordinate $a$ is the residual height of the active core-side
branch above the mining anchor, and the second coordinate $b$ is the
residual height of the Frontier-side branch. The state $(a,0)$ has no
positive Frontier-side suffix. A state $(a,b)$ with $0\le a<b$ is a strict
catch-up state; its tie-origin label is necessarily minus and is therefore
omitted. A state $(a,b)^\sigma$ with $a>b\ge1$ is a strict lead state whose
current tie-origin label is $\sigma$.

The state $D_b^\sigma$ is a private diagonal state of height $b$ with
tie-origin label $\sigma$. It is not a public tie: it records private
equality, possibly before the core-side branch is released. Private
block generation does not change $\sigma$. Whenever the maximum public
height increases, the successor label is plus if the private height is at
least the new public height and minus otherwise, as specified in
Section~\ref{sec:model-real}.

Public tie states are denoted by $E_{b,r}^{\sigma}$ for
$\sigma\in\{-,+\}$.
Here $b$ is the length of the tied public branches, $\sigma$ is the
tie-origin label carried into the tie, and $r$ is the hidden reserve length
behind the core-side tied public branch.

In each canonical state, the \emph{active branch} is the unique core-side
branch that extends, and hence does not conflict with, the current mining
anchor. Every other core-side branch is \emph{non-active}.

The definition of a canonical state requires every non-active core-side
branch to have height strictly below the current mining anchor.

\begin{lemma}
\label{lem:canonical-state-representation}
\label{lem:branch-separation-invariant}
Every reachable IDEAL one-deviation state has a unique representation as an
element of $\widehat{\mathcal X}$. In this representation, core-side suffixes
and reserves consist only of $m_i$-blocks, while Frontier-side suffixes
consist only of Frontier-aggregate blocks.
\end{lemma}

\begin{proof}
We induct on the number of primitive transitions. The initial state is
$(0,0)$, so the claim is immediate. Assume that the current state has a
canonical representation with branch separation. We check its state type
and the feasible canonical actions.

At every state, $\mathsf{NoRelease}$ leaves the state unchanged and therefore
preserves both the canonical representation and branch separation. It remains
to check the other canonical actions.

\emph{Initial and strict catch-up states.}
At $(0,0)$, $\mathsf{Hedge}(0)$ leads to $(1,0)$ after an $m_i$-success
and to $(0,1)$ after a Frontier success. Under $\mathsf{NoMining}$, the
state remains $(0,0)$ and the core selects another action with probability
$p$; with probability $q$, it waits until a Frontier success gives
$(0,1)$.

At a strict catch-up state $(a,b)$ with $1\le a<b$, a feasible
$\mathsf{Wait}$ gives
\[
  (a,b)
  \longrightarrow
  \begin{cases}
    D_b^-,
      & m_i\text{-success and }a+1=b,\\
    (a+1,b),
      & m_i\text{-success and }a+1<b,\\
    (a,b+1),
      & \text{Frontier success}.
  \end{cases}
\]
At every strict catch-up state, including $(0,b)$, a feasible
$\mathsf{Hedge}(d)$ gives
\[
  (a,b)
  \longrightarrow
  \begin{cases}
    (1,0), & m_i\text{-success and }d=0,\\
    D_1^-, & m_i\text{-success and }d=1,\\
    (1,d), & m_i\text{-success and }d\ge2,\\
    (a,b+1), & \text{Frontier success}.
  \end{cases}
\]
On an $m_i$-success, the new core-side suffix is the single new
$m_i$-block and the length-$d$ Frontier-side suffix contains only
Frontier-aggregate blocks. If an old core-side branch exists, it becomes
non-active and, by $\mathsf{Hedge}$ feasibility, lies strictly below the
successor mining anchor, as required by the canonical-state definition. On a
Frontier success, only the Frontier-side suffix grows.

Under $\mathsf{NoMining}$, the state remains $(a,b)$ and the core selects
another action with probability $p$. With probability $q$, it waits until
a state change gives $(a,b+1)$.

\emph{Lead states.}
At $(a,0)$ with $a\ge1$, $\mathsf{Wait}$ gives $(a+1,0)$ after an
$m_i$-success. A Frontier success gives $D_1^+$ if $a=1$ and
$(a,1)^+$ if $a>1$. Under $\mathsf{NoMining}$, the state remains $(a,0)$
on its probability-$p$ branch and has the same state-changing successor as
$\mathsf{Wait}$ on its probability-$q$ branch.

At a signed lead state $(a,b)^\sigma$ with $a>b\ge1$,
$\mathsf{Wait}$ gives $(a+1,b)^\sigma$ after an $m_i$-success. A Frontier
success gives $D_{b+1}^+$ if $a=b+1$ and $(a,b+1)^+$ if $a>b+1$.
Under $\mathsf{NoMining}$, the state remains $(a,b)^\sigma$ on its
probability-$p$ branch and has the same state-changing successor on its
probability-$q$ branch.

Writing $a=b+m$, $\mathsf{Override}(j)$ from either kind of lead state
leaves $(m-j,0)$. From $(b+r,b)^\sigma$ with $b,r\ge1$,
$\mathsf{Match}$ creates $E_{b,r}^\sigma$. These release actions preserve
the all-$m_i$ core-side suffix or reserve.

\emph{Private diagonal states.}
From $D_b^\sigma$, $\mathsf{Wait}$ gives $(b+1,b)^\sigma$ after an
$m_i$-success and $(b,b+1)$ after a Frontier success. Under
$\mathsf{NoMining}$, the state remains $D_b^\sigma$ and the core selects
another action with probability $p$; with probability $q$, it waits until
a state change gives $(b,b+1)$. The action $\mathsf{Match}$ gives
$E_{b,0}^\sigma$. These transitions preserve branch separation.

\emph{Public tie states.}
From $E_{b,r}^\sigma$, an $m_i$-success under $\mathsf{TieWait}$ gives
$E_{b,r+1}^\sigma$. A Frontier success on the Frontier-side branch gives
\[
  \begin{cases}
    (b,b+1), & r=0,\\
    D_{b+1}^+, & r=1,\\
    (b+r,b+1)^+, & r\ge2,
  \end{cases}
\]
while a Frontier success on the core-side tied branch gives
\[
  \begin{cases}
    (0,1), & r=0,\\
    D_1^+, & r=1,\\
    (r,1)^+, & r\ge2.
  \end{cases}
\]
Under $\mathsf{NoMining}$, the state remains $E_{b,r}^\sigma$ and the core
selects another action with probability $p$. With probability $q$, it
waits until a state change gives the two outcomes above, with
conditional probabilities $1-\gamma^\sigma$ and $\gamma^\sigma$,
respectively. Under $\mathsf{TieWait}$, the $m_i$-success appends an
$m_i$-block to the reserve. In either Frontier
outcome, the new block is the only new Frontier-aggregate block on the
relevant Frontier-side suffix. Finally, $\mathsf{TieOverride}(j)$ leaves
the residual state $(r-j,0)$. Thus every successor again has branch
separation.

The state classes in $\widehat{\mathcal X}$ are disjoint, and their numerical
coordinates and origin labels are uniquely determined by the current canonical
tree. Hence the representation is unique.
\end{proof}

For each $x\in\widehat{\mathcal X}$, let
$\widehat{\mathcal A}(x)$ denote the canonical actions defined in
Section~\ref{sec:model-ideal} that are feasible at $x$.

\subsection{Exact Centered Bellman Inequalities}
\label{subsec:centered-one-step-inequalities}

For each realized canonical transition from a reduced state $x$, under a
feasible canonical action $u$, to a successor state $x'$, let
\[
  R(x,u,x')
\]
denote the reward credited to $m_i$ for blocks newly included in the common
trunk by that transition, as determined by their reward-owner labels, and let
\[
  L(x,u,x')
\]
denote the public-frontier advancement caused by that transition.
The quantity $L$ does not measure advancement of the common trunk or the
mining anchor.

Let
\[
  V:\widehat{\mathcal X}\to\mathbb R_{\ge0}
\]
be a nonnegative potential normalized by
\[
  V(0,0)=0.
\]
For each feasible state--action pair $(x,u)$, where
$x\in\widehat{\mathcal X}$ and $u\in\widehat{\mathcal A}(x)$, let $X'$ be
the successor state and define the exact centered action value by
\[
\begin{aligned}
  \mathcal B_V(x,u)
  &=
  \mathbb E\!\left[
    R(x,u,X')-pL(x,u,X')
  \right.\\
  &\qquad\left.
    {}+V(X')\mid x,u
  \right].
\end{aligned}
\]
The exact centered Bellman inequality is the one-step condition
\[
  V(x)\ge\mathcal B_V(x,u)
  \tag{3.1}\label{eq:3-1}
\]
for every reduced state $x\in\widehat{\mathcal X}$ and every feasible
canonical action $u\in\widehat{\mathcal A}(x)$.

The inequality says that the potential at the current state covers the
expected centered reward $R-pL$ from the next transition together with the
potential remaining at the successor state.

\begin{lemma}
\label{lem:ideal-policy-bound}
Suppose that $V\ge0$, $V(0,0)=0$, and $V$ satisfies the exact centered
Bellman inequality~\textup{(\ref{eq:3-1})} for every reduced canonical
state $x\in\widehat{\mathcal X}$ and every feasible canonical action
$u\in\widehat{\mathcal A}(x)$. Then every IDEAL policy has long-run utility
\[
  U_i\le p.
\]
\end{lemma}

\begin{proof}
Fix an IDEAL policy $\pi$. Let $(\mathcal F_t)_{t\ge0}$ be the natural
filtration, let $H_t$ be the IDEAL history through transition $t$, let
$Y_t=\operatorname{cur}(H_t)$, and let $X_t$ be the canonical reduced
representation of $Y_t$. Set $U_t=\pi(H_t)$. Write
\[
  R_t=R(X_t,U_t,X_{t+1}),
  \qquad
  L_t=L(X_t,U_t,X_{t+1}).
\]
Since $U_t$ is $\mathcal F_t$-measurable, applying
\textup{(\ref{eq:3-1})} conditionally on $\mathcal F_t$ gives
\[
  \mathbb E\!\left[
    R_t-pL_t+V(X_{t+1})
    \mid \mathcal F_t
  \right]
  \le V(X_t).
  \tag{3.2}\label{eq:3-2}
\]
Taking expectations, summing over $t=0,\ldots,n-1$, and canceling the
intermediate potential terms yields
\[
  \mathbb E\!\left[
    \sum_{t=0}^{n-1}(R_t-pL_t)
  \right]
  \le
  V(X_0)-\mathbb E[V(X_n)]
  \le 0,
  \tag{3.3}\label{eq:3-3}
\]
where the last inequality uses $X_0=(0,0)$, $V(0,0)=0$, and $V\ge0$.

Define the cumulative reward and cumulative public-frontier advancement by
\[
  \mathcal R_n:=\sum_{t=0}^{n-1}R_t,
  \qquad
  \mathcal L_n:=\sum_{t=0}^{n-1}L_t.
  \tag{3.4}\label{eq:3-4}
\]
Equation~\textup{(\ref{eq:3-3})} is therefore equivalent to
\[
  \mathbb E[\mathcal R_n]
  \le
  p\,\mathbb E[\mathcal L_n]
  \tag{3.5}\label{eq:3-5}
\]
for every transition depth $n$. Whenever
$\mathbb E[\mathcal L_n]>0$, this is equivalently
\[
  \frac{\mathbb E[\mathcal R_n]}
       {\mathbb E[\mathcal L_n]}
  \le p.
  \tag{3.6}\label{eq:3-6}
\]
Taking the $\liminf$ in~\textup{(\ref{eq:3-6})} gives $U_i\le p$.
\end{proof}

\subsection{Conservative Action Values}
\label{subsec:action-values}

The action values below are conservative upper bounds on the exact
centered action values defined in
Section~\ref{subsec:centered-one-step-inequalities}.
For each feasible state--action pair $(x,u)$ with
$u\in\widehat{\mathcal A}(x)$, let
$\overline{\mathcal B}_V(x,u)$ denote the corresponding conservative
action value. It is chosen so that
\[
  \mathcal B_V(x,u)
  \le
  \overline{\mathcal B}_V(x,u).
\]
Consequently, if
\[
  V(x)\ge\overline{\mathcal B}_V(x,u)
\]
for every reduced state and feasible canonical action, then the exact
centered Bellman inequality~\textup{(\ref{eq:3-1})} also holds.

For a potential $V$, we write these conservative centered action values
as follows. For $\sigma\in\{-,+\}$, set
\[
  \gamma^\sigma
  =
  \begin{cases}
    \gamma^-, & \sigma=-,\\
    \gamma^+, & \sigma=+.
  \end{cases}
\]

For origin-sensitive lead and equality successors, write
\[
  V^\sigma(a,b)
  =
  \begin{cases}
    V(D_b^\sigma), & a=b\ge1,\\
    V((a,b)^\sigma), & a>b\ge1,\\
    V(0,0), & a=b=0,\\
    V(a,b), & \text{otherwise}.
  \end{cases}
\]
Thus the superscript records the tie origin of a signed lead or private
diagonal state; it never denotes a public tie.

\paragraph{\textsf{Wait}.}
At a strict catch-up state $(a,b)$ with $1\le a<b$, Wait has conservative
centered value
\[
  W_V(a,b)
  =
  pV^-(a+1,b)+qV(a,b+1)-pq.
  \tag{3.7}\label{eq:3-7}
\]
At a lead state with no positive Frontier-side suffix, the value is
\[
  W_V(a,0)
  =
  pV(a+1,0)+qV^+(a,1)-pq,
  \qquad a\ge1.
\]
At a signed lead state $(a,b)^\sigma$ with $a>b\ge1$, the value is
\[
  W_V^\sigma(a,b)
  =
  pV^\sigma(a+1,b)+qV^+(a,b+1)-pq.
\]

In a private diagonal state $D_b^\sigma$, where $\sigma\in\{-,+\}$, Wait has conservative centered value
\[
  W_{D,V}^\sigma(b)
  =
  pV^\sigma(b+1,b)+qV(b,b+1)-pq.
  \tag{3.8}\label{eq:3-8}
\]
The core-success branch retains $\sigma$, whereas the Frontier-success
branch enters the unsigned strict catch-up region.

\paragraph{\textsf{Override}.}
Suppose the state has lead $m$:
\[
  a=b+m,
  \qquad
  m\ge1.
\]
For $b=0$ the state is $(a,0)$; for $b\ge1$ it is $(a,b)^\sigma$.
For $1\le j\le m$, $\mathrm{Override}(j)$ releases $b+j$ active-branch blocks, makes the core-side branch uniquely longest, and leaves residual lead $m-j$.  Its conservative centered value is
\[
  O_{j,V}(b+m,b)
  =
  V(m-j,0)+b+j-p.
  \tag{3.9}\label{eq:3-9}
\]
When $j=m$, the successor is $(0,0)$, so \textup{(\ref{eq:3-9})} uses $V(0,0)=0$.

\paragraph{\textsf{Match}.}
From private diagonal states, Match creates a public tie with no immediate
reward and no conservative advancement charge:
\[
  T^\sigma_{D,V}(b)
  =
  V(E_{b,0}^\sigma),
  \qquad
  \sigma\in\{-,+\}.
  \tag{3.10}\label{eq:3-10}
\]
Equivalently,
\[
  D_b^-
  \xrightarrow{\mathrm{Match}}
  E_{b,0}^-,
  \qquad
  D_b^+
  \xrightarrow{\mathrm{Match}}
  E_{b,0}^+.
\]

From a signed lead state
\[
  (b+r,b)^\sigma,
  \qquad
  b\ge1,\quad r\ge1,
\]
Match creates a public tie carrying the current origin and leaves reserve
$r$:
\[
  T^\sigma_{L,V}(b,r)
  =
  V(E_{b,r}^\sigma).
  \tag{3.11}\label{eq:3-11}
\]
The ensuing public-tie dynamics are handled by later transitions.

\paragraph{\textsf{Hedge}.}
For $\mathrm{Hedge}(d)$, define the successful reanchoring successor
\[
  H_d
  =
  \begin{cases}
    (1,0), & d=0,\\
    D_1^-, & d=1,\\
    (1,d), & d\ge2.
  \end{cases}
\]
At the initial state $(0,0)$, $\mathrm{Hedge}(0)$ is feasible, with conservative centered value
\[
  H_{0,V}(0,0)
  =
  pV(H_0)+qV(0,1)-pq.
\]
At a strict catch-up state $(a,b)$ with $0\le a<b$, $\mathrm{Hedge}(d)$ is feasible for
\[
  0\le d\le b
  \quad\text{if }a=0,
  \qquad
  0\le d<b-a
  \quad\text{if }a\ge1.
\]
Its conservative centered value is
\[
  H_{d,V}(a,b)
  =
  pV(H_d)+qV(a,b+1)-pq.
  \tag{3.12}\label{eq:3-12}
\]
In this conservative accounting, the miner-success branch receives no
immediate reward or advancement charge. The Frontier-success branch
extends the public branch by one and contributes the term $-pq$.

\paragraph{\textsf{TieWait}.}
At a public tie state $E_{b,r}^\sigma$, TieWait mines on $m_i$'s tied branch or reserve tip and withholds any newly discovered block.  Its conservative centered value is
\[
  \begin{aligned}
  W^\sigma_{E,V}(b,r)
  &=
  pV(E_{b,r+1}^\sigma)
  +\gamma^\sigma q\{b+V^+(r,1)\}  \\
  &\qquad
  +(1-\gamma^\sigma)qV^+(b+r,b+1)-pq.
  \end{aligned}
\tag{3.13}\label{eq:3-13}
\]
The value $V^+(r,1)$ is used because, when equality occurs at $(1,1)$, the equality is plus-origin.  Likewise, if $b+r=b+1$, then $V^+(b+r,b+1)=V(D_{b+1}^+)$.

\paragraph{\textsf{NoMining}.}
At the initial state or a strict catch-up state $(a,b)$, $\mathsf{NoMining}$
has conservative centered value
\[
  N_V(a,b)
  =
  pV(a,b)+q\{V(a,b+1)-p\}.
\]
At a lead state $(a,0)$, its value is
\[
  N_V(a,0)
  =
  pV(a,0)+q\{V^+(a,1)-p\},
  \qquad a\ge1.
\]
At a signed lead state $(a,b)^\sigma$, its value is
\[
  N_V^\sigma(a,b)
  =
  pV((a,b)^\sigma)+q\{V^+(a,b+1)-p\}.
\]
At a private diagonal state $D_b^\sigma$, its value is
\[
  N_{D,V}^\sigma(b)
  =
  pV(D_b^\sigma)+q\{V(b,b+1)-p\}.
\]
At a public tie state $E_{b,r}^\sigma$, its conservative centered value is
\[
\begin{aligned}
  N_{E,V}^\sigma(b,r)
  &={}
  pV(E_{b,r}^\sigma)\\
  &\quad
  {}+q\Bigl[
    \gamma^\sigma\{b+V^+(r,1)\}\\
  &\hspace{5em}
    +(1-\gamma^\sigma)V^+(b+r,b+1)-p
  \Bigr].
\end{aligned}
\]
The first term is the outcome in which the core is selected, the state remains
unchanged, and the core chooses again. The bracketed term is the conditional
value when the core waits until the state changes. For the core-side tie resolution,
fixing $b$ core blocks while advancing the public frontier by one contributes
$b-p$; the competing resolution contributes $-p$.

\paragraph{\textsf{NoRelease}.}
At every reduced state $x$, $\mathsf{NoRelease}$ has exact and conservative
centered value
\[
  \overline{\mathcal B}_V(x,\mathsf{NoRelease})=V(x).
\]
\paragraph{\textsf{TieOverride}.}
At $E_{b,r}^\sigma$, the action $\mathrm{TieOverride}(j)$ immediately publishes $j$ existing reserve blocks and breaks the public tie in favor of $m_i$'s branch.  The feasible range is
\[
  1\le j\le r.
\]
Its conservative centered value is
\[
  Q^\sigma_{j,V}(b,r)
  =
  V(r-j,0)+b+j-p.
  \tag{3.14}\label{eq:3-14}
\]

\subsection{Potential Family}
\label{subsec:potential-family}

Fix integers
\[
  N\ge1,
  \qquad
  D\ge1.
\]
We construct a finite-dimensional family of potentials.  The parameter $N$ controls the finite core depth, and $D$ controls the number of explicit deficit tails.

For $1\le d\le D+1$, define
\[
  B_d=\max\{N+1,d+1\}.
\]
The lower bound $B_d$ is the smallest value of $b$ for which a strict catch-up tail state with deficit $d=b-a$ has both $b>N$ and $a\ge1$.

We also define
\[
  K=\max\{2,N-D\}.
\]
The terminal coefficients $w_1,\ldots,w_K$ are included as LP variables.  If $D\ge N$, then $K=2$, and the terminal part uses only $w_1,w_2$ explicitly.

\paragraph{Capitulation states.}
For all $b\ge0$, set
\[
  V(0,b)=0.
  \tag{3.15}\label{eq:3-15}
\]

\paragraph{Lead states.}
For $a\ge b+1$, define the affine lead potential
\[
  L(a,b)=\lambda a-\mu b-\kappa,
  \tag{3.16}\label{eq:3-16}
\]
where $\lambda,\mu,\kappa$ are LP variables. We assign this potential to
the unsigned lead states with no positive Frontier-side suffix and to both
origins of every signed lead state:
\[
  V(a,0)=L(a,0)
  \qquad(a\ge1),
\]
and
\[
  V((a,b)^\sigma)=L(a,b)
  \qquad(a>b\ge1,\ \sigma\in\{-,+\}).
\]
Thus the signed lead states add no LP variables.

\paragraph{Finite strict catch-up core.}
For
\[
  1\le a<b\le N,
\]
set
\[
  V(a,b)=S_{a,b},
  \tag{3.17}\label{eq:3-17}
\]
where each $S_{a,b}$ is an LP variable.

\paragraph{Private diagonal states.}
For $1\le b\le N$, introduce LP variables
\[
  V(D_b^-)=D_b^-,
  \qquad
  V(D_b^+)=D_b^+.
  \tag{3.18}\label{eq:3-18}
\]
These are private equal decision states, not public tie states.

For $b>N$, set
\[
  V(D_b^-)=D^-(b)=g^-b+h^-,
  \tag{3.19a}
\]
and
\[
  V(D_b^+)=D^+(b)=g^+b+h^+,
  \tag{3.19b}
\]
where $g^-,h^-,g^+,h^+$ are LP variables.

\paragraph{Public tie states with no reserve.}
For $1\le b\le N$, introduce LP variables only for reserve-zero public tie states:
\[
  V(E_{b,0}^-)=M_b,
  \qquad
  V(E_{b,0}^+)=P_b.
  \tag{3.20}\label{eq:3-20}
\]
For $b>N$, set
\[
  V(E_{b,0}^-)=E_0^-(b)=e^-b+f^-,
  \tag{3.21a}
\]
and
\[
  V(E_{b,0}^+)=E_0^+(b)=e^+b+f^+,
  \tag{3.21b}
\]
where $e^-,f^-,e^+,f^+$ are LP variables.

\paragraph{Public tie states with reserve.}
For every $b\ge1$, $r\ge1$, and $\sigma\in\{-,+\}$, define
\[
  V(E_{b,r}^\sigma)=L(b+r,b).
  \tag{3.22}\label{eq:3-22}
\]
Thus reserve-carrying public tie states introduce no additional LP variables.  In particular, Match from a strict lead state
\[
  (b+r,b)^\sigma\to E_{b,r}^\sigma
\]
has the same potential value on both sides:
\[
  V((b+r,b)^\sigma)=L(b+r,b)=V(E_{b,r}^\sigma).
\]
The Bellman inequality for this Match transition is therefore automatically satisfied as equality.

\paragraph{Affine strict catch-up tails.}
Let
\[
  d=b-a
\]
be the deficit.  For $b>N$ and $1\le d\le D$, set
\[
  V(a,b)=C_d(b)=u_db+v_d,
  \tag{3.23}\label{eq:3-23}
\]
where $u_d,v_d$ are LP variables.

\paragraph{Terminal superharmonic tail.}
For $b>N$ and $d>D$, write
\[
  d=D+k.
\]
For $1\le k\le K$, set
\[
  V(a,b)=a w_k.
  \tag{3.24}\label{eq:3-24}
\]
For $k>K$, the remaining sequence
\[
  w_{K+1},w_{K+2},\ldots
\]
is not included in the finite LP.  Its existence is guaranteed later by the terminal extension lemma.

\subsection{Shorthand for Boundary Values}
\label{subsec:boundary-shorthand}

We use the following shorthand in the LP constraints.

For private diagonal states, define
\[
  \mathcal D^\sigma(b)
  =
  \begin{cases}
    D_b^\sigma, & 1\le b\le N,\\
    g^\sigma b+h^\sigma, & b>N,
  \end{cases}
  \qquad
  \sigma\in\{-,+\}.
  \tag{3.25}\label{eq:3-25}
\]
Here $(g^\sigma,h^\sigma)$ means $(g^-,h^-)$ for $\sigma=-$ and $(g^+,h^+)$ for $\sigma=+$.

For reserve-zero public tie states, define
\[
  \mathcal E_0^-(b)
  =
  \begin{cases}
    M_b, & 1\le b\le N,\\
    e^-b+f^-, & b>N,
  \end{cases}
  \tag{3.26a}
\]
and
\[
  \mathcal E_0^+(b)
  =
  \begin{cases}
    P_b, & 1\le b\le N,\\
    e^+b+f^+, & b>N.
  \end{cases}
  \tag{3.26b}
\]
For all public tie states, including reserve-carrying ones, define
\[
  \mathcal E^\sigma(b,r)
  =
  \begin{cases}
    \mathcal E_0^\sigma(b), & r=0,\\
    L(b+r,b), & r\ge1.
  \end{cases}
  \tag{3.27}\label{eq:3-27}
\]

For unsigned states and the common potential of signed lead states appearing
in finite constraints, define
\[
  \mathcal V(a,b)
  =
  \begin{cases}
    0, & a=0,\\
    L(a,b), & a\ge b+1,\\
    S_{a,b}, & 1\le a<b\le N,\\
    C_{b-a}(b), & b>N,\ 1\le b-a\le D,\\
    a w_{b-a-D}, & b>N,\ D<b-a\le D+K.
  \end{cases}
  \tag{3.28}\label{eq:3-28}
\]
For $a>b\ge1$, $\mathcal V(a,b)=L(a,b)$ is the common value assigned to
both $(a,b)^-$ and $(a,b)^+$. When $a=b>0$, the appropriate value is the
private diagonal value $\mathcal D^-(b)$ or $\mathcal D^+(b)$.

For origin-sensitive successor values, define
\[
  \mathcal V^\sigma(a,b)
  =
  \begin{cases}
    0, & a=b=0,\\
    \mathcal D^\sigma(b), & a=b\ge1,\\
    \mathcal V(a,b), & \text{otherwise},
  \end{cases}
  \qquad
  \sigma\in\{-,+\}.
  \tag{3.29}\label{eq:3-29}
\]
Thus $\mathcal V^\sigma(a,b)$ selects the origin-specific diagonal value
when $a=b$ and the common signed-lead potential when $a>b\ge1$.

The value of a successful Hedge successor is abbreviated by
\[
  \mathcal H_d
  =
  \begin{cases}
    \mathcal V(1,0), & d=0,\\
    \mathcal D^-(1), & d=1,\\
    \mathcal V(1,d), & d\ge2.
  \end{cases}
  \tag{3.30}\label{eq:3-30}
\]
Here $\mathcal V(1,0)=L(1,0)$.

Finally, define
\[
  X_b=\mathcal V(b,b+1).
  \tag{3.31}\label{eq:3-31}
\]
Thus $X_b$ is the strict catch-up value reached when the Frontier aggregate extends a private diagonal state of length $b$, or when the Frontier side wins a public tie of length $b$ with no remaining hidden reserve.

The choice
\[
  K=\max\{2,N-D\}
\]
ensures that every terminal value appearing in the finite-core boundary constraints is represented among $w_1,\ldots,w_K$.  Values with terminal index larger than $K$ are handled by the terminal extension lemma rather than by additional LP variables.

\subsection{Certificate LP Constraints}
\label{sec:ideal-certificate-lp}

For fixed $p,\gamma^-,\gamma^+$, all constraints below are linear.  The certificate LP is denoted by
\[
  \mathcal L_{N,D}(p,\gamma^-,\gamma^+).
\]
As before, let
\[
  q=1-p.
\]

Whenever an affine slack $\Delta(b)$ is required to be nonnegative for all $b\ge b_0$, we impose the two linear constraints
\[
  \operatorname{slope}_b\Delta(b)\ge0,
  \qquad
  \Delta(b_0)\ge0.
\]
The same convention is used for affine scaling slacks.

The LP contains Bellman inequalities and, for several constraint
families, additional prefix-scaling constraints. For an
inequality slack $\Delta$ with auxiliary slack $\Gamma$, the corresponding
scaling slack is
\[
  \Xi=\Delta+p\Gamma\ge0.
\]
This prefix-scaling constraint is stronger than what is needed for
soundness at a fixed $p$, but it is useful for the prefix-monotonicity
argument.

\paragraph{C0: Nonnegativity and terminal bounds.}

For all finite strict catch-up states,
\[
  S_{a,b}\ge0
  \qquad
  (1\le a<b\le N).
  \tag{C0.1}\label{constr:c0-1}
\]

For all finite private diagonal states,
\[
  D_b^-\ge0,
  \qquad
  D_b^+\ge0
  \qquad
  (1\le b\le N).
  \tag{C0.2}\label{constr:c0-2}
\]

For all finite reserve-zero public tie states,
\[
  M_b\ge0,
  \qquad
  P_b\ge0
  \qquad
  (1\le b\le N).
  \tag{C0.3}\label{constr:c0-3}
\]

For every affine strict catch-up tail $1\le d\le D$,
\[
  u_d\ge0,
  \tag{C0.4}\label{constr:c0-4}
\]
and
\[
  u_dB_d+v_d\ge0.
  \tag{C0.5}\label{constr:c0-5}
\]

For the private diagonal tails,
\[
  g^-\ge0,
  \qquad
  g^-(N+1)+h^-\ge0,
  \tag{C0.6}\label{constr:c0-6}
\]
and
\[
  g^+\ge0,
  \qquad
  g^+(N+1)+h^+\ge0.
  \tag{C0.7}\label{constr:c0-7}
\]

For the reserve-zero public tie tails,
\[
  e^-\ge0,
  \qquad
  e^-(N+1)+f^-\ge0,
  \tag{C0.8}\label{constr:c0-8}
\]
and
\[
  e^+\ge0,
  \qquad
  e^+(N+1)+f^+\ge0.
  \tag{C0.9}\label{constr:c0-9}
\]

For the terminal coefficients,
\[
  0\le w_k\le q
  \qquad
  (1\le k\le K).
  \tag{C0.10}\label{constr:c0-10}
\]
For Hedge monotonicity in the terminal region, impose
\[
  w_k\ge w_{k+1}
  \qquad
  (1\le k<K).
  \tag{C0.11}\label{constr:c0-11}
\]

\paragraph{C1: Lead constraints.}

The lead coefficients satisfy
\[
  \lambda-\kappa=q,
  \tag{C1.1}\label{constr:c1-1}
\]
\[
  \lambda\ge1,
  \tag{C1.2}\label{constr:c1-2}
\]
\[
  \mu\ge0,
  \tag{C1.3}\label{constr:c1-3}
\]
\[
  \lambda-\mu\ge1,
  \tag{C1.4}\label{constr:c1-4}
\]
and
\[
  q\mu-p\lambda+pq\ge0.
  \tag{C1.5}\label{constr:c1-5}
\]

These constraints imply nonnegativity of the lead potential for $(a,0)$
and for both signed states $(a,b)^\sigma$ with $a>b\ge1$. They also imply
all Override inequalities from these lead states, all Match inequalities
from signed lead states, all TieOverride inequalities from reserve-carrying
public tie states, and the TieWait inequalities at reserve $r\ge2$.

\paragraph{C2: Strict catch-up Wait constraints.}

First consider the finite strict catch-up core.  For every
\[
  1\le a<b\le N,
\]
define
\[
  Y_{a,b}=\mathcal V^-(a+1,b),
  \qquad
  Z_{a,b}=\mathcal V(a,b+1).
\]
Equivalently, if $a+1=b$, then
\[
  Y_{a,b}=\mathcal D^-(b),
\]
and otherwise $Y_{a,b}=\mathcal V(a+1,b)$.

Impose the Bellman inequality
\[
  S_{a,b}
  \ge
  pY_{a,b}+qZ_{a,b}-pq,
  \tag{C2.1}\label{constr:c2-1}
\]

For affine strict catch-up tails, let $1\le d\le D$ and $b\ge B_d$.  Define
\[
  Y_d(b)
  =
  \begin{cases}
    \mathcal D^-(b), & d=1,\\
    C_{d-1}(b), & 2\le d\le D,
  \end{cases}
\]
and
\[
  Z_d(b)
  =
  \begin{cases}
    C_{d+1}(b+1), & 1\le d<D,\\
    (b-D)w_1, & d=D.
  \end{cases}
\]
The affine Wait slack is
\[
  \Delta^W_d(b)
  =
  C_d(b)-\{pY_d(b)+qZ_d(b)-pq\}.
  \tag{C2.2}\label{constr:c2-2}
\]
Impose
\[
  \operatorname{slope}_b\Delta^W_d(b)\ge0,
  \tag{C2.3}\label{constr:c2-3}
\]
and
\[
  \Delta^W_d(B_d)\ge0.
  \tag{C2.4}\label{constr:c2-4}
\]

The boundary successor for $d=1$ is the minus-origin private diagonal value $\mathcal D^-(b)$, not a public tie value.

\paragraph{C3: Private diagonal constraints.}

For every $1\le b\le N$, impose the private diagonal Wait constraints
\[
  D_b^-
  \ge
  pL(b+1,b)+qX_b-pq,
  \tag{C3.1}\label{constr:c3-1}
\]
and
\[
  D_b^+
  \ge
  pL(b+1,b)+qX_b-pq.
  \tag{C3.2}\label{constr:c3-2}
\]
Here an $m_i$-success from $D_b^\sigma$ reaches $(b+1,b)^\sigma$,
whose potential is $L(b+1,b)$ for either origin; a state-changing
Frontier discovery reaches the unsigned catch-up state $(b,b+1)$.
The associated prefix-scaling constraint is
\[
  D_b^\sigma\ge X_b
  \qquad
  (1\le b\le N,\ \sigma\in\{-,+\}).
  \tag{C3.3}\label{constr:c3-3}
\]

For every $1\le b\le N$, impose the private diagonal Match constraints
\[
  D_b^-\ge M_b,
  \tag{C3.4}\label{constr:c3-4}
\]
and
\[
  D_b^+\ge P_b.
  \tag{C3.5}\label{constr:c3-5}
\]
These encode
\[
  D_b^-
  \xrightarrow{\mathrm{Match}}
  E_{b,0}^-,
  \qquad
  D_b^+
  \xrightarrow{\mathrm{Match}}
  E_{b,0}^+.
\]

For the private diagonal Wait tails $b\ge N+1$ and $\sigma\in\{-,+\}$, define
\[
  \Delta^\sigma_{D,W}(b)
  =
  \mathcal D^\sigma(b)
  -
  \{pL(b+1,b)+qC_1(b+1)-pq\}.
  \tag{C3.6}\label{constr:c3-6}
\]
Impose
\[
  \operatorname{slope}_b\Delta^\sigma_{D,W}(b)\ge0,
  \tag{C3.7}\label{constr:c3-7}
\]
and
\[
  \Delta^\sigma_{D,W}(N+1)\ge0.
  \tag{C3.8}\label{constr:c3-8}
\]

The corresponding tail scaling slack is
\[
  \Xi^\sigma_{D,W}(b)
  =
  \mathcal D^\sigma(b)-C_1(b+1).
  \tag{C3.9}\label{constr:c3-9}
\]
Impose
\[
  \operatorname{slope}_b\Xi^\sigma_{D,W}(b)\ge0,
  \tag{C3.10}\label{constr:c3-10}
\]
and
\[
  \Xi^\sigma_{D,W}(N+1)\ge0.
  \tag{C3.11}\label{constr:c3-11}
\]

For the private diagonal Match tails, define
\[
  \Delta^-_{D,M}(b)
  =
  \mathcal D^-(b)-\mathcal E^-_0(b),
  \tag{C3.12}\label{constr:c3-12}
\]
and
\[
  \Delta^+_{D,M}(b)
  =
  \mathcal D^+(b)-\mathcal E^+_0(b).
  \tag{C3.13}\label{constr:c3-13}
\]
Impose
\[
  \operatorname{slope}_b\Delta^-_{D,M}(b)\ge0,
  \qquad
  \operatorname{slope}_b\Delta^+_{D,M}(b)\ge0,
  \tag{C3.14}\label{constr:c3-14}
\]
and
\[
  \Delta^-_{D,M}(N+1)\ge0,
  \qquad
  \Delta^+_{D,M}(N+1)\ge0.
  \tag{C3.15}\label{constr:c3-15}
\]

\paragraph{C4: Public tie constraints.}

Public tie states are decision states. The LP explicitly imposes TieWait
constraints for reserve $r=0$ and reserve $r=1$. The reserve $r\ge2$
TieWait constraints follow from C1, and TieOverride constraints are implied
by C1. NoMining requires no additional constraints: its inequalities follow
from the existing C1, C3, C4, and C7 constraints.

First consider reserve-zero public tie states $E_{b,0}^\sigma$.  For every $1\le b\le N$ and $\sigma\in\{-,+\}$, impose
\[
  \mathcal E^\sigma_0(b)
  \ge
  pL(b+1,b)
  +\gamma^\sigma qb
  +(1-\gamma^\sigma)qX_b
  -pq.
  \tag{C4.1}\label{constr:c4-1}
\]
The associated prefix-scaling constraint is
\[
  \mathcal E^\sigma_0(b)
  \ge
  \gamma^\sigma qb+(1-\gamma^\sigma)X_b.
  \tag{C4.2}\label{constr:c4-2}
\]

For the reserve-zero public tie tails $b\ge N+1$, define
\[
\begin{aligned}
\Delta^\sigma_{E0,W}(b)
&=
\mathcal E^\sigma_0(b)\\
&\quad
-
\Bigl\{
  pL(b+1,b)+\gamma^\sigma qb\\
&\qquad
  +(1-\gamma^\sigma)qC_1(b+1)-pq
\Bigr\}.
\end{aligned}
\tag{C4.3}\label{constr:c4-3}
\]
Impose
\[
  \operatorname{slope}_b\Delta^\sigma_{E0,W}(b)\ge0,
  \tag{C4.4}\label{constr:c4-4}
\]
and
\[
  \Delta^\sigma_{E0,W}(N+1)\ge0.
  \tag{C4.5}\label{constr:c4-5}
\]

The corresponding tail scaling slack is
\[
  \Xi^\sigma_{E0,W}(b)
  =
  \mathcal E^\sigma_0(b)
  -\gamma^\sigma qb
  -(1-\gamma^\sigma)C_1(b+1).
  \tag{C4.6}\label{constr:c4-6}
\]
Impose
\[
  \operatorname{slope}_b\Xi^\sigma_{E0,W}(b)\ge0,
  \tag{C4.7}\label{constr:c4-7}
\]
and
\[
  \Xi^\sigma_{E0,W}(N+1)\ge0.
  \tag{C4.8}\label{constr:c4-8}
\]

Next consider reserve-one public tie states $E_{b,1}^\sigma$.  Since
\[
  V(E_{b,1}^\sigma)=L(b+1,b),
\]
TieWait gives, for every $1\le b\le N$ and $\sigma\in\{-,+\}$,
\[
\begin{aligned}
L(b+1,b)
&\ge
pL(b+2,b)
+\gamma^\sigma q\{b+\mathcal D^+(1)\}\\
&\quad
+(1-\gamma^\sigma)q\mathcal D^+(b+1)-pq.
\end{aligned}
\tag{C4.9}\label{constr:c4-9}
\]
The associated prefix-scaling constraint is
\[
  L(b+1,b)
  \ge
  \gamma^\sigma qb
  +\gamma^\sigma \mathcal D^+(1)
  +(1-\gamma^\sigma)\mathcal D^+(b+1).
  \tag{C4.10}\label{constr:c4-10}
\]

For the reserve-one public tie tails $b\ge N+1$, define
\[
  \begin{aligned}
  \Delta^\sigma_{E1,W}(b)
  &=
  L(b+1,b)\\
  &\quad
  -
  \Bigl\{
    pL(b+2,b)
    +\gamma^\sigma q\{b+\mathcal D^+(1)\}\\
  &\qquad
    +(1-\gamma^\sigma)q\mathcal D^+(b+1)
    -pq
  \Bigr\}.
  \end{aligned}
\tag{C4.11}\label{constr:c4-11}
\]
Impose
\[
  \operatorname{slope}_b\Delta^\sigma_{E1,W}(b)\ge0,
  \tag{C4.12}\label{constr:c4-12}
\]
and
\[
  \Delta^\sigma_{E1,W}(N+1)\ge0.
  \tag{C4.13}\label{constr:c4-13}
\]

The corresponding tail scaling slack is
\[
\begin{aligned}
\Xi^\sigma_{E1,W}(b)
&=
L(b+1,b)\\
&\quad
-
\Bigl\{
  \gamma^\sigma qb
  +\gamma^\sigma \mathcal D^+(1)
  +(1-\gamma^\sigma)\mathcal D^+(b+1)
\Bigr\}.
\end{aligned}
\tag{C4.14}\label{constr:c4-14}
\]
Impose
\[
  \operatorname{slope}_b\Xi^\sigma_{E1,W}(b)\ge0,
  \tag{C4.15}\label{constr:c4-15}
\]
and
\[
  \Xi^\sigma_{E1,W}(N+1)\ge0.
  \tag{C4.16}\label{constr:c4-16}
\]

\paragraph{C5: Hedge-success upper constraints.}

The successful Hedge successor $H_d$ must satisfy
\[
  V(H_d)\le q.
  \tag{C5}
\]
The case $d=0$ is
\[
  V(H_0)=V(1,0)=L(1,0)=q,
\]
which follows from \ref{constr:c1-1}.  For $d=1$, impose
\[
  D_1^-\le q.
  \tag{C5.1}\label{constr:c5-1}
\]
For $2\le b\le N$, impose
\[
  S_{1,b}\le q.
  \tag{C5.2}\label{constr:c5-2}
\]
For affine-tail Hedge successors, equivalently states $(1,d+1)$ with $d+1>N$ and $1\le d\le D$, impose
\[
  C_d(d+1)\le q.
  \tag{C5.3}\label{constr:c5-3}
\]
For terminal Hedge successors, no additional constraint is needed because \ref{constr:c0-10} gives $w_k\le q$.

\paragraph{C6: Hedge monotonicity constraints.}

By C5, for every feasible $\mathrm{Hedge}(d)$,
\[
  pV(H_d)+qV(a,b+1)-pq
  \le
  qV(a,b+1).
\]
Since $V\ge0$, it is enough to impose the stronger monotonicity condition
\[
  V(a,b)\ge V(a,b+1)
\]
throughout the strict catch-up region.

For every finite strict catch-up state $1\le a<b\le N$, impose
\[
  S_{a,b}\ge \mathcal V(a,b+1).
  \tag{C6.1}\label{constr:c6-1}
\]

For affine strict catch-up tails, let $1\le d\le D$ and $b\ge B_d$.  Define
\[
  Z^H_d(b)
  =
  \begin{cases}
    C_{d+1}(b+1), & 1\le d<D,\\
    (b-D)w_1, & d=D.
  \end{cases}
\]
The affine Hedge-monotonicity slack is
\[
  \Delta^H_d(b)=C_d(b)-Z^H_d(b).
  \tag{C6.2}\label{constr:c6-2}
\]
Impose
\[
  \operatorname{slope}_b\Delta^H_d(b)\ge0,
  \tag{C6.3}\label{constr:c6-3}
\]
and
\[
  \Delta^H_d(B_d)\ge0.
  \tag{C6.4}\label{constr:c6-4}
\]
In the terminal region, Hedge monotonicity is guaranteed by \ref{constr:c0-11} and Lemma~\ref{lem:terminal-extension}.

\paragraph{C7: One-block lead Wait constraints.}

The lead-state Wait constraints with lead at least two are implied by C1.
The remaining cases are $(1,0)$ and, for $b\ge1$, the two signed
one-block lead states
\[
  (b+1,b)^\sigma,
  \qquad \sigma\in\{-,+\}.
\]

For every $0\le b\le N$, impose
\[
  L(b+1,b)
  \ge
  pL(b+2,b)+q\mathcal D^+(b+1)-pq.
  \tag{C7.1}\label{constr:c7-1}
\]
Here a state-changing Frontier discovery reaches $D_{b+1}^+$ from
$(1,0)$ when $b=0$ and from either $(b+1,b)^\sigma$ when $b\ge1$.
The constraint is independent of the incoming sign because both signed
lead states have potential $L(b+1,b)$.

The associated prefix-scaling constraint is
\[
  L(b+1,b)\ge \mathcal D^+(b+1)
  \qquad
  (0\le b\le N).
  \tag{C7.2}\label{constr:c7-2}
\]

For the one-block lead tail $b\ge N+1$, define
\[
  \Delta_{1,W}(b)
  =
  L(b+1,b)
  -
  \{pL(b+2,b)+q\mathcal D^+(b+1)-pq\}.
  \tag{C7.3}\label{constr:c7-3}
\]
Impose
\[
  \operatorname{slope}_b\Delta_{1,W}(b)\ge0,
  \tag{C7.4}\label{constr:c7-4}
\]
and
\[
  \Delta_{1,W}(N+1)\ge0.
  \tag{C7.5}\label{constr:c7-5}
\]

The corresponding scaling slack is
\[
  \Xi_{1,W}(b)
  =
  L(b+1,b)-\mathcal D^+(b+1).
  \tag{C7.6}\label{constr:c7-6}
\]
Impose
\[
  \operatorname{slope}_b\Xi_{1,W}(b)\ge0,
  \tag{C7.7}\label{constr:c7-7}
\]
and
\[
  \Xi_{1,W}(N+1)\ge0.
  \tag{C7.8}\label{constr:c7-8}
\]

\paragraph{C8: Terminal strict catch-up constraints.}

At deficit $d=D+1$, write
\[
  a=b-(D+1).
\]
The terminal boundary Wait slack is
\[
  \Delta_{\mathrm{bd}}(b)
  =
  aw_1
  -
  \{pC_D(b)+qaw_2-pq\}.
  \tag{C8.1}\label{constr:c8-1}
\]
It must hold for every $b\ge B_{D+1}$.  Since it is affine in $b$, impose
\[
  \operatorname{slope}_b\Delta_{\mathrm{bd}}(b)\ge0,
  \tag{C8.2}\label{constr:c8-2}
\]
and
\[
  \Delta_{\mathrm{bd}}(B_{D+1})\ge0.
  \tag{C8.3}\label{constr:c8-3}
\]

For the explicit terminal prefix, impose
\[
  w_k-pw_{k-1}-qw_{k+1}\ge0
  \qquad
  (2\le k\le K-1).
  \tag{C8.4}\label{constr:c8-4}
\]
If $K=2$, this family is empty.

Finally, impose the terminal extension lower condition
\[
  w_K\ge \frac{p}{q}w_{K-1}.
  \tag{C8.5}\label{constr:c8-5}
\]
Together with \ref{constr:c0-10} and \ref{constr:c0-11}, this condition permits an infinite terminal extension satisfying the required superharmonicity and monotonicity constraints.

\subsection{Terminal Extension}
\label{subsec:terminal-extension}

\begin{lemma}
\label{lem:terminal-extension}
Let $p<1/2$, $q=1-p$, and
\[
  \rho=\frac{p}{q}.
\]
Suppose $w_{K-1},w_K$ satisfy
\[
  0\le w_{K-1}\le q,
  \qquad
  0\le w_K\le q,
\]
\[
  w_{K-1}\ge w_K,
  \qquad
  w_K\ge \rho w_{K-1}.
\]
Then there exists an infinite sequence
\[
  w_{K+1},w_{K+2},\ldots
\]
such that, for every $k\ge K$,
\[
  0\le w_k\le q,
\]
\[
  w_k\ge w_{k+1},
\]
and
\[
  w_k-pw_{k-1}-qw_{k+1}\ge0.
\]
\end{lemma}

\begin{proof}
Re-index the prescribed pair as
\[
  z_1=w_{K-1},
  \qquad
  z_2=w_K.
\]
We construct an infinite sequence
\[
  z_1,z_2,z_3,\ldots
\]
and then set
\[
  w_{K+\ell-1}=z_{\ell+1}
  \qquad
  (\ell\ge1).
\]

Define
\[
  A=\frac{z_2-\rho z_1}{1-\rho},
  \qquad
  B=\frac{z_1-z_2}{1-\rho},
\]
and set
\[
  z_\ell=A+B\rho^{\ell-1}
  \qquad
  (\ell\ge1).
\]
Since $z_2\ge \rho z_1$, we have $A\ge0$.  Since $z_1\ge z_2$, we have $B\ge0$.  The definition gives the prescribed first two terms:
\[
  A+B=z_1,
  \qquad
  A+B\rho=z_2.
\]

Because $0<\rho<1$ and $B\ge0$, the sequence is nonincreasing:
\[
  z_\ell\ge z_{\ell+1}
  \qquad
  (\ell\ge1).
\]
It is also nonnegative and bounded above by $z_1\le q$, so
\[
  0\le z_\ell\le q
  \qquad
  (\ell\ge1).
\]

Finally, for every $\ell\ge2$,
\[
\begin{aligned}
  z_\ell-pz_{\ell-1}-qz_{\ell+1}
  &=
  A(1-p-q)\\
  &\quad
  +B\rho^{\ell-2}\{\rho-p-q\rho^2\}.
\end{aligned}
\]
The first term is zero because $p+q=1$.  The second term is also zero because $\rho=p/q$, and hence
\[
  \rho-p-q\rho^2=0.
\]
Therefore
\[
  z_\ell-pz_{\ell-1}-qz_{\ell+1}=0
  \qquad
  (\ell\ge2).
\]

Translating the sequence back to the $w$-notation gives an infinite extension with
\[
\begin{gathered}
0\le w_k\le q,\qquad
w_k\ge w_{k+1},\\
w_k-pw_{k-1}-qw_{k+1}\ge0
\end{gathered}
\]
for every $k\ge K$.  This proves the lemma.
\end{proof}

\subsection{Soundness}
\label{subsec:lp-soundness}

\begin{theorem}
\label{thm:ideal-lp-soundness}
Fix $p<1/2$ and $\gamma^-,\gamma^+\in[0,1]$. If
\[
  \mathcal L_{N,D}(p,\gamma^-,\gamma^+)
\]
is feasible, then the potential $V$ constructed from the feasible solution
satisfies the exact centered Bellman inequality for every reduced state and
every feasible canonical action. Consequently, no IDEAL policy can obtain
long-run utility exceeding $p$.
\end{theorem}

\begin{proof}
Take a feasible solution of
\[
  \mathcal L_{N,D}(p,\gamma^-,\gamma^+).
\]
Construct $V$ from the potential family in
Section~\ref{subsec:potential-family}, use the shorthand of
Section~\ref{subsec:boundary-shorthand}, and extend the terminal sequence
using Lemma~\ref{lem:terminal-extension}. By
Section~\ref{subsec:action-values}, it suffices to verify the conservative
Bellman inequalities
\[
  V(x)
  \ge
  \overline{\mathcal B}_V(x,u)
\]
for every reduced state $x$ and every feasible canonical action
$u\in\widehat{\mathcal A}(x)$. These inequalities imply the exact
centered Bellman inequality~\textup{(\ref{eq:3-1})}.

\paragraph{1. Nonnegativity and the initial state.}

First, $V(0,0)=0$ by construction.  We also have $V(x)\ge0$ for every reduced state $x$.

For finite strict catch-up states, finite private diagonal states, and finite reserve-zero public tie states, nonnegativity follows from (\ref{constr:c0-1})--(\ref{constr:c0-3}).  For affine strict catch-up tails it follows from (\ref{constr:c0-4})--(\ref{constr:c0-5}).  For private diagonal tails it follows from (\ref{constr:c0-6})--(\ref{constr:c0-7}), and for reserve-zero public tie tails it follows from (\ref{constr:c0-8})--(\ref{constr:c0-9}).  For terminal tails it follows from (\ref{constr:c0-10}) and Lemma~\ref{lem:terminal-extension}.

It remains only to check lead states and reserve-carrying public tie states.
Both are governed by the lead potential. Write $a=b+m$, $m\ge1$. For
$b=0$ the lead state is $(a,0)$; for $b\ge1$ it is $(a,b)^\sigma$. In
either case its potential is
\[
  L(b+m,b)
  =
  b(\lambda-\mu)+m\lambda-\kappa.
\]
Using (\ref{constr:c1-1})--(\ref{constr:c1-4}), we have
\[
  \lambda-\kappa=q,
  \qquad
  \lambda\ge1,
  \qquad
  \lambda-\mu\ge1.
\]
Hence the minimum over $b\ge0$ and $m\ge1$ occurs at $b=0,m=1$, where
\[
  L(1,0)=\lambda-\kappa=q\ge0.
\]
Thus all lead states are nonnegative.  Since reserve-carrying public tie states satisfy
\[
  V(E^\sigma_{b,r})=L(b+r,b)
  \qquad
  (r\ge1),
\]
they are also nonnegative.

At the initial state, the only relevant public-mining transition is the canonical $\mathrm{Hedge}(0)$ transition.  Since
\[
  V(H_0)=V(1,0)=L(1,0)=q
\]
and
\[
  V(0,1)=0,
\]
its conservative centered value is
\[
  pV(1,0)+qV(0,1)-pq
  =
  pq-pq
  =
  0
  =
  V(0,0).
\]
Thus the Bellman inequality holds at the initial state.

\paragraph{2. Wait from finite strict catch-up states.}

Let
\[
  1\le a<b\le N.
\]
The Wait value is
\[
  pY_{a,b}+qZ_{a,b}-pq,
\]
where
\[
  Y_{a,b}=\mathcal V^-(a+1,b),
  \qquad
  Z_{a,b}=\mathcal V(a,b+1).
\]
Constraint (\ref{constr:c2-1}) states exactly
\[
  S_{a,b}
  \ge
  pY_{a,b}+qZ_{a,b}-pq.
\]
Since $V(a,b)=S_{a,b}$ on the finite strict catch-up core, the Bellman inequality holds.

The boundary case $a+1=b$ is handled by
\[
  Y_{a,b}=\mathcal D^-(b),
\]
so a miner discovery enters the minus-origin private diagonal state $D_b^-$, not a public tie state.

\paragraph{3. Wait from affine strict catch-up tails.}

Let
\[
  b>N,
  \qquad
  1\le d=b-a\le D,
  \qquad
  a\ge1.
\]
Then $b\ge B_d$, and the Wait slack is
\[
  \Delta^W_d(b)
  =
  C_d(b)-\{pY_d(b)+qZ_d(b)-pq\}.
\]
By (\ref{constr:c2-3})--(\ref{constr:c2-4}), this affine slack is nonnegative for every $b\ge B_d$.  Therefore the Wait Bellman inequality holds throughout every affine strict catch-up tail.

For $d=1$, the miner-success successor is the minus-origin private diagonal tail value
\[
  \mathcal D^-(b),
\]
not a public tie value.

\paragraph{4. Wait from terminal strict catch-up tails.}

Let
\[
  b>N,
  \qquad
  d=b-a>D,
  \qquad
  d=D+k.
\]

First consider the terminal boundary $k=1$, so that
\[
  a=b-(D+1).
\]
The Wait slack is
\[
  \Delta_{\mathrm{bd}}(b)
  =
  aw_1
  -
  \{pC_D(b)+qaw_2-pq\}.
\]
By (\ref{constr:c8-2})--(\ref{constr:c8-3}), this affine slack is nonnegative for every
\[
  b\ge B_{D+1}.
\]
Hence the Bellman inequality holds at the boundary between the explicit affine-tail region and the terminal region.

Now consider $k\ge2$.  In the terminal region,
\[
  V(a,b)=aw_k.
\]
If $m_i$ discovers the next block under Wait, then the successor is
\[
  (a+1,b)
\]
and has value
\[
  (a+1)w_{k-1}.
\]
If the Frontier aggregate discovers the next block, then the successor is
\[
  (a,b+1)
\]
and has value
\[
  aw_{k+1}.
\]
The Wait slack is therefore
\[
\begin{aligned}
&aw_k-\{p(a+1)w_{k-1}+qaw_{k+1}-pq\}  \\
&\qquad
=
p(q-w_{k-1})
+
a\{w_k-pw_{k-1}-qw_{k+1}\}.
\end{aligned}
\]
The first term is nonnegative because $w_{k-1}\le q$, by (\ref{constr:c0-10}) and Lemma~\ref{lem:terminal-extension}.  The second term is nonnegative by (\ref{constr:c8-4}) when $2\le k<K$, and by Lemma~\ref{lem:terminal-extension} when $k\ge K$.  Thus the Bellman inequality holds for all terminal strict catch-up states.

\paragraph{5. Wait from finite private diagonal states.}

Let
\[
  1\le b\le N,
  \qquad
  \sigma\in\{-,+\}.
\]
The $m_i$-success from $D_b^\sigma$ reaches $(b+1,b)^\sigma$ and
retains the sign; both signed successors have potential $L(b+1,b)$.
The Frontier-success branch reaches $(b,b+1)$. Hence the Wait value is
\[
  pL(b+1,b)+qX_b-pq.
\]
For $\sigma=-$, constraint (\ref{constr:c3-1}) gives
\[
  D_b^-
  \ge
  pL(b+1,b)+qX_b-pq.
\]
For $\sigma=+$, constraint (\ref{constr:c3-2}) gives
\[
  D_b^+
  \ge
  pL(b+1,b)+qX_b-pq.
\]
Since
\[
  V(D_b^-)=D_b^-,
  \qquad
  V(D_b^+)=D_b^+,
\]
the Bellman inequality holds for Wait from every finite private diagonal state.

\paragraph{6. Wait from private diagonal tails.}

Let
\[
  b\ge N+1,
  \qquad
  \sigma\in\{-,+\}.
\]
The private diagonal Wait slack is
\[
  \Delta^\sigma_{D,W}(b)
  =
  \mathcal D^\sigma(b)
  -
  \{pL(b+1,b)+qC_1(b+1)-pq\}.
\]
By (\ref{constr:c3-7})--(\ref{constr:c3-8}), this affine slack is nonnegative for every $b\ge N+1$.  Thus the Bellman inequality holds for Wait from all private diagonal tail states.

\paragraph{7. Wait from one-block lead states.}

A one-block lead state is $(1,0)$ when $b=0$ and
$(b+1,b)^\sigma$ when $b\ge1$. Both signed states are covered by the
same constraint because their potential is $L(b+1,b)$.
For $0\le b\le N$, constraint (\ref{constr:c7-1}) gives
\[
  L(b+1,b)
  \ge
  pL(b+2,b)+q\mathcal D^+(b+1)-pq.
\]
This is exactly the Wait Bellman inequality, because a state-changing
Frontier discovery reaches the plus-origin private diagonal state
\[
  D^+_{b+1}.
\]

For $b\ge N+1$, the one-block lead Wait slack is
\[
  \Delta_{1,W}(b)
  =
  L(b+1,b)
  -
  \{pL(b+2,b)+q\mathcal D^+(b+1)-pq\}.
\]
By (\ref{constr:c7-4})--(\ref{constr:c7-5}), this affine slack is nonnegative for every $b\ge N+1$. Thus the Bellman inequality holds for Wait from every one-block lead state.

\paragraph{8. Wait from lead states with lead at least two.}

Let $m\ge2$. The state is $(m,0)$ when $b=0$ and
$(b+m,b)^\sigma$ when $b\ge1$. An $m_i$-success retains $\sigma$ in
the signed region, whereas a state-changing Frontier discovery assigns
plus origin. Both successors have the same affine lead potential. Using
\[
  L(a,b)=\lambda a-\mu b-\kappa,
\]
the Wait slack is
\[
\begin{aligned}
&L(b+m,b)
-
\{pL(b+m+1,b)\\
&\qquad
+qL(b+m,b+1)-pq\}  \\
&
=
q\mu-p\lambda+pq.
\end{aligned}
\]
This is nonnegative by (\ref{constr:c1-5}).  Therefore the Bellman inequality holds for Wait from every lead state with lead at least two.

\paragraph{9. Override from lead states.}

Let $m\ge1$. The state is $(m,0)$ when $b=0$ and
$(b+m,b)^\sigma$ when $b\ge1$; the calculation is independent of
$\sigma$.
For $1\le j\le m$, the Override value is
\[
  V(m-j,0)+b+j-p.
\]

First suppose $j<m$.  Then the successor $(m-j,0)$ is a lead state, and
\[
  V(m-j,0)=L(m-j,0).
\]
The Override slack is
\[
\begin{aligned}
&L(b+m,b)-\{L(m-j,0)+b+j-p\} \\
&\qquad
=
b(\lambda-\mu-1)+j(\lambda-1)+p.
\end{aligned}
\]
This is nonnegative by (\ref{constr:c1-2}) and (\ref{constr:c1-4}).

Now suppose $j=m$.  Then the successor is $(0,0)$, and the Override slack is
\[
\begin{aligned}
&L(b+m,b)-(b+m-p) \\
&\qquad
=
b(\lambda-\mu-1)+m(\lambda-1)+p-\kappa.
\end{aligned}
\]
The minimum over $b\ge0$ and $m\ge1$ occurs at $b=0,m=1$, where it equals
\[
  \lambda-\kappa+p-1
  =
  q+p-1
  =
  0
\]
by (\ref{constr:c1-1}).  Hence every Override action from every lead state satisfies the Bellman inequality.

\paragraph{10. Match from private diagonal states.}

A Match action from a private diagonal state creates the corresponding
reserve-zero public tie state and records no immediate reward or
conservative advancement charge:
\[
  D_b^-
  \xrightarrow{\mathrm{Match}}
  E_{b,0}^-,
  \qquad
  D_b^+
  \xrightarrow{\mathrm{Match}}
  E_{b,0}^+.
\]
For $1\le b\le N$, constraints (\ref{constr:c3-4})--(\ref{constr:c3-5}) give
\[
  D_b^-\ge M_b,
  \qquad
  D_b^+\ge P_b.
\]
Thus the Bellman inequality holds for Match from finite private diagonal states.

For $b\ge N+1$, the Match slacks are
\[
  \Delta^-_{D,M}(b)
  =
  \mathcal D^-(b)-\mathcal E^-_0(b),
\]
and
\[
  \Delta^+_{D,M}(b)
  =
  \mathcal D^+(b)-\mathcal E^+_0(b).
\]
By (\ref{constr:c3-14})--(\ref{constr:c3-15}), both affine slacks are nonnegative for every $b\ge N+1$.  Hence the Bellman inequality holds for Match from all private diagonal tail states.

\paragraph{11. Match from strict lead states.}

Let the signed lead state be
\[
  (b+r,b)^\sigma,
  \qquad
  b\ge1,
  \qquad
  r\ge1.
\]
A Match action creates the reserve-carrying public tie state
\[
  E^\sigma_{b,r}.
\]
By definition of the potential on reserve-carrying public tie states,
\[
  V(E^\sigma_{b,r})
  =
  L(b+r,b)
  =
  V((b+r,b)^\sigma).
\]
The conservative action value records no immediate reward or advancement
charge at the Match step. Therefore the conservative Bellman inequality
holds as equality.

\paragraph{12. Hedge actions.}

Consider a strict catch-up state $(a,b)$ with $0\le a<b$, and let $\mathrm{Hedge}(d)$ be feasible. Thus $0\le d\le b$ if $a=0$, and $0\le d<b-a$ if $a\ge1$. Its conservative centered value is
\[
  pV(H_d)+qV(a,b+1)-pq.
\]
By C5,
\[
  V(H_d)\le q.
\]
Hence
\[
  pV(H_d)+qV(a,b+1)-pq
  \le
  qV(a,b+1).
\]
For $a=0$, we have
\[
  V(0,b)=V(0,b+1)=0,
\]
so the Bellman inequality holds.
For $1\le a<b$, C6 imposes
\[
  V(a,b)\ge V(a,b+1)
\]
through the finite-core, affine-tail, and terminal monotonicity constraints. Since $V\ge0$ and $q\le1$,
\[
  V(a,b)\ge V(a,b+1)\ge qV(a,b+1).
\]
Thus every strict-catch-up $\mathrm{Hedge}(d)$ action satisfies the Bellman inequality.

\paragraph{13. TieWait from reserve-zero public tie states.}

Let
\[
  E^\sigma_{b,0},
  \qquad
  \sigma\in\{-,+\}.
\]
The two state-changing successors are $(0,1)$ and $(b,b+1)$, both in the
unsigned strict catch-up region.
For $1\le b\le N$, the TieWait Bellman inequality is precisely (\ref{constr:c4-1}):
\[
  \mathcal E^\sigma_0(b)
  \ge
  pL(b+1,b)
  +\gamma^\sigma qb
  +(1-\gamma^\sigma)qX_b
  -pq.
\]

For $b\ge N+1$, the reserve-zero TieWait slack is
\[
\begin{aligned}
  \Delta^\sigma_{E0,W}(b)
  &=
  \mathcal E^\sigma_0(b)
  \\
  &\quad -
  \Bigl\{
    pL(b+1,b)
    +\gamma^\sigma qb
    \\
  &\qquad
    +(1-\gamma^\sigma)qC_1(b+1)
    -pq
  \Bigr\}.
\end{aligned}
\]
By (\ref{constr:c4-4})--(\ref{constr:c4-5}), this affine slack is nonnegative for every $b\ge N+1$. Thus the Bellman inequality holds for TieWait from all reserve-zero public tie states.

\paragraph{14. TieWait from reserve-one public tie states.}

Let
\[
  E^\sigma_{b,1},
  \qquad
  \sigma\in\{-,+\}.
\]
The public-height update assigns plus origin to both state-changing
successors, $D_1^+$ and $D_{b+1}^+$.
Since
\[
  V(E^\sigma_{b,1})=L(b+1,b),
\]
the finite reserve-one TieWait inequality is exactly (\ref{constr:c4-9}):
\[
\begin{aligned}
  L(b+1,b)
  &\ge
  pL(b+2,b)
  \\
  &\quad
  +\gamma^\sigma q\{b+\mathcal D^+(1)\}
  \\
  &\quad
  +(1-\gamma^\sigma)q\mathcal D^+(b+1)
  -pq
\end{aligned}
\]
for $1\le b\le N$.

For $b\ge N+1$, the reserve-one TieWait slack is
\[
\begin{aligned}
  \Delta^\sigma_{E1,W}(b)
  &=
  L(b+1,b)\\
  &\quad
  -
  \Bigl\{
    pL(b+2,b)
    +\gamma^\sigma q\{b+\mathcal D^+(1)\}\\
  &\qquad
    +(1-\gamma^\sigma)q\mathcal D^+(b+1)
    -pq
  \Bigr\}.
\end{aligned}
\]
By (\ref{constr:c4-12})--(\ref{constr:c4-13}), this affine slack is nonnegative for every $b\ge N+1$. Thus the Bellman inequality holds for TieWait from all reserve-one public tie states.

\paragraph{15. TieWait from public tie states with reserve at least two.}

Let
\[
  E^\sigma_{b,r},
  \qquad
  r\ge2.
\]
Both state-changing successors are plus-origin lead states,
$(r,1)^+$ and $(b+r,b+1)^+$.
The potential is
\[
  V(E^\sigma_{b,r})=L(b+r,b).
\]
The conservative TieWait centered value is
\[
\begin{aligned}
  &pL(b+r+1,b)
  +\gamma^\sigma q\{b+L(r,1)\} \\
  &\qquad
  +(1-\gamma^\sigma)qL(b+r,b+1)
  -pq.
\end{aligned}
\]
Subtracting this value from $L(b+r,b)$ gives
\[
  b\gamma^\sigma q(\lambda-\mu-1)
  +
  (q\mu-p\lambda+pq).
\]
The first term is nonnegative by (\ref{constr:c1-4}), and the second term is nonnegative by (\ref{constr:c1-5}).  Therefore the Bellman inequality holds for TieWait from all public tie states with reserve at least two.

\paragraph{16. TieOverride actions.}

At a public tie state $E^\sigma_{b,r}$, the feasible range for $\mathrm{TieOverride}(j)$ is
\[
  1\le j\le r.
\]
There is no TieOverride action when $r=0$.

The TieOverride value is
\[
  V(r-j,0)+b+j-p.
\]
Since
\[
  V(E^\sigma_{b,r})=L(b+r,b),
\]
the TieOverride slack is
\[
  L(b+r,b)-\{V(r-j,0)+b+j-p\}.
\]
This is exactly the Override slack from the lead state
\[
  (b+r,b)^\sigma
\]
with override parameter $j$.  By paragraph 9, it is nonnegative.  Therefore every TieOverride action satisfies the Bellman inequality.

\paragraph{17. NoRelease actions.}
For $\mathsf{NoRelease}$, the successor is the current state and both reward
and public-frontier advancement are zero. Its Bellman inequality is therefore
the identity
\[
  V(x)\ge V(x).
\]

\paragraph{18. NoMining actions.}
Every NoMining value has the form
\[
  pV(x)+qC_V(x),
\]
where $C_V(x)$ is the conditional value obtained by waiting until a
state change. Since $p<1/2$, we have $q>0$, and
\[
  V(x)\ge pV(x)+qC_V(x)
  \quad\Longleftrightarrow\quad
  V(x)\ge C_V(x).
\]
It therefore suffices to bound the conditional state-changing value.

At a capitulation state $(0,b)$, including $(0,0)$, this value is
\[
  V(0,b+1)-p=-p\le0=V(0,b).
\]
At a strict catch-up state, C6 gives
\[
  V(a,b)\ge V(a,b+1),
\]
and at a private diagonal state the finite constraint
\ref{constr:c3-3} and the tail constraints
\ref{constr:c3-9}--\ref{constr:c3-11} give
\[
  V(D_b^\sigma)\ge V(b,b+1).
\]
These inequalities are stronger than the corresponding NoMining
conditional inequalities, whose right-hand sides subtract $p$.

For $(1,0)$ and the signed one-block lead states $(b+1,b)^\sigma$, the
same conclusion follows from
\ref{constr:c7-2} in the finite region and
\ref{constr:c7-6}--\ref{constr:c7-8} in the tail. For a lead of at least
two blocks, C1 gives
\[
  L(a,b)-L(a,b+1)=\mu\ge0,
\]
which again implies the NoMining inequality for either incoming origin;
the state-changing successor has plus origin but the lead potential is
the same for both signs.

For the public-tie cases, the finite and tail prefix-scaling constraints
\ref{constr:c3-3}, \ref{constr:c3-9}--\ref{constr:c3-11},
\ref{constr:c7-2}, and \ref{constr:c7-6}--\ref{constr:c7-8}, together
with C1, give
\[
  X_b
  \le \mathcal D^+(b)
  \le L(b,b-1)
  \le L(b+1,b).
\]
The last inequality follows from $\lambda-\mu\ge1$, and
\[
  L(b+1,b)=b(\lambda-\mu)+q\ge b+q.
\]
At $E_{b,0}^\sigma$, write
\[
  \mathcal N_0^\sigma(b)
  =
  \gamma^\sigma b+(1-\gamma^\sigma)X_b-p.
\]
The preceding bounds imply $\mathcal N_0^\sigma(b)\le L(b+1,b)$. The TieWait constraint
\ref{constr:c4-1}, or its tail version
\ref{constr:c4-3}--\ref{constr:c4-5}, therefore gives
\[
  V(E_{b,0}^\sigma)
  \ge pL(b+1,b)+q\mathcal N_0^\sigma(b)
  \ge \mathcal N_0^\sigma(b).
\]

At $E_{b,1}^\sigma$, let
\[
  \mathcal N_1^\sigma(b)
  =
  \gamma^\sigma\{b+\mathcal D^+(1)\}
  +(1-\gamma^\sigma)\mathcal D^+(b+1)-p.
\]
Constraint \ref{constr:c7-2} at zero and C1 give
\[
  \mathcal D^+(1)\le L(1,0)=q,
\]
while the finite or tail C7 constraints give
\[
  \mathcal D^+(b+1)\le L(b+1,b)<L(b+2,b).
\]
Moreover, C1 implies $b+q\le L(b+2,b)$. Hence
$\mathcal N_1^\sigma(b)\le L(b+2,b)$, and the TieWait constraint
\ref{constr:c4-9}, or its tail version
\ref{constr:c4-11}--\ref{constr:c4-13}, yields
\[
  V(E_{b,1}^\sigma)
  =L(b+1,b)
  \ge pL(b+2,b)+q\mathcal N_1^\sigma(b)
  \ge \mathcal N_1^\sigma(b).
\]

Finally, let $r\ge2$. Then
\[
  V(E_{b,r}^\sigma)=L(b+r,b),
\]
and both state-changing successors are plus-origin lead states. The
conditional slack is
\[
\begin{aligned}
&L(b+r,b)
-\Bigl[
  \gamma^\sigma\{b+L(r,1)\}\\
&\hspace{7em}
  +(1-\gamma^\sigma)L(b+r,b+1)-p
\Bigr]\\
&\qquad={}
\gamma^\sigma b(\lambda-\mu-1)+\mu+p
\ge0,
\end{aligned}
\]
where the final inequality uses $\lambda-\mu\ge1$ and $\mu\ge0$ from C1.
Therefore $C_V(x)\le V(x)$ at every state.

Thus every feasible canonical state-action pair satisfies its conservative
Bellman inequality. Section~\ref{subsec:action-values} then gives the
exact centered Bellman inequality~\textup{(\ref{eq:3-1})} for every pair.

By Lemma~\ref{lem:ideal-policy-bound}, every IDEAL policy has long-run
utility at most $p$.

\end{proof}

\subsection{Monotonicity in Hash Share}
\label{subsec:monotonicity-hash-share}

LP feasibility is prefix-monotone in the tested miner's hash share.  This
property justifies threshold search over $p$.

\begin{theorem}
\label{thm:prefix-monotonicity}
Fix $(N,D)$ and $\gamma^-,\gamma^+$.  If the certificate LP
\[
  \mathcal L_{N,D}(c,\gamma^-,\gamma^+)
\]
is feasible at some $c<1/2$, then
\[
  \mathcal L_{N,D}(p,\gamma^-,\gamma^+)
\]
is feasible at every $p<c$.
\end{theorem}

\begin{proof}
Let
\[
  q_c=1-c,
  \qquad
  q_p=1-p,
\]
and define
\[
  \theta=\frac{q_p}{q_c}>1.
\]
Take a feasible solution of
\[
  \mathcal L_{N,D}(c,\gamma^-,\gamma^+).
\]
Let $V_c$ be the corresponding potential, with terminal sequence extended by Lemma~\ref{lem:terminal-extension}.  Define
\[
  V_p=\theta V_c.
\]
Equivalently, scale every LP variable by the same factor $\theta$:
\[
  \begin{gathered}
  S^p_{a,b}=\theta S^c_{a,b},
  \qquad
  (D_b^\sigma)^p=\theta (D_b^\sigma)^c,
  \\
  M_b^p=\theta M_b^c,
  \qquad
  P_b^p=\theta P_b^c,
  \\
  u_d^p=\theta u_d^c,
  \qquad
  v_d^p=\theta v_d^c,
  \\
  (g^\sigma)^p=\theta (g^\sigma)^c,
  \qquad
  (h^\sigma)^p=\theta (h^\sigma)^c,
  \\
  (e^\sigma)^p=\theta (e^\sigma)^c,
  \qquad
  (f^\sigma)^p=\theta (f^\sigma)^c,
  \\
  w_k^p=\theta w_k^c,
  \\
  \lambda_p=\theta\lambda_c,
  \qquad
  \mu_p=\theta\mu_c,
  \qquad
  \kappa_p=\theta\kappa_c,
  \end{gathered}
  \tag{3.32}\label{eq:3-32}
\]
for $\sigma\in\{-,+\}$.  Capitulation values remain
\[
  V_p(0,b)=0.
\]
Reserve-carrying public tie values also scale automatically, because for $r\ge1$,
\[
  V_p(E^\sigma_{b,r})
  =
  L_p(b+r,b)
  =
  \theta L_c(b+r,b)
  =
  \theta V_c(E^\sigma_{b,r}).
\]

We verify that the scaled variables satisfy every constraint of
\[
  \mathcal L_{N,D}(p,\gamma^-,\gamma^+).
\]

\paragraph{C0: Nonnegativity and terminal bounds.}

All nonnegativity constraints are preserved because $\theta>0$.  For example,
\[
  u_d^pB_d+v_d^p
  =
  \theta(u_d^cB_d+v_d^c),
\]
and the same argument applies to finite core variables, private diagonal tails, reserve-zero public tie tails, and terminal coefficients.

The terminal upper bounds are preserved because
\[
  0\le w_k^c\le q_c
  \quad\Longrightarrow\quad
  0\le w_k^p=\theta w_k^c\le \theta q_c=q_p.
\]
The terminal monotonicity constraints are homogeneous:
\[
  w_k^p-w_{k+1}^p
  =
  \theta(w_k^c-w_{k+1}^c)\ge0.
\]

The terminal extension lower condition is also preserved.  Since the function $x\mapsto x/(1-x)$ is increasing on $(0,1)$ and $p<c$, we have
\[
  \frac{p}{q_p}<\frac{c}{q_c}.
\]
Hence
\[
  w_K^p
  =
  \theta w_K^c
  \ge
  \theta\frac{c}{q_c}w_{K-1}^c
  =
  \frac{c}{q_c}w_{K-1}^p
  \ge
  \frac{p}{q_p}w_{K-1}^p.
\]

\paragraph{C1: Lead constraints.}

The lead normalization is preserved:
\[
  \lambda_p-\kappa_p
  =
  \theta(\lambda_c-\kappa_c)
  =
  \theta q_c
  =
  q_p.
\]
The inequalities
\[
  \lambda\ge1,
  \qquad
  \mu\ge0,
  \qquad
  \lambda-\mu\ge1
\]
are preserved because $\theta>1$.

It remains to check the lead mining constraint.  Let
\[
  \Delta_{\mathrm{lead}}^t
  =
  q_t\mu_t-t\lambda_t+tq_t
\]
be its slack at hash share $t$. Direct substitution gives
\[
\begin{aligned}
\Delta_{\mathrm{lead}}^p
&=
\theta\{
  \Delta_{\mathrm{lead}}^c
  +(c-p)(\lambda_c+\mu_c-q_c)
\}.
\end{aligned}
\]
The scaling slack is
\[
  \Delta_{\mathrm{lead}}^c
  +c(\lambda_c+\mu_c-q_c)
  =
  \mu_c,
\]
which is nonnegative by \ref{constr:c1-3}.  Therefore \ref{constr:c1-5} is preserved by the
generic calculation below.

\paragraph{Generic prefix scaling.}

We use the following calculation repeatedly. For each one-step constraint
family, let $\Delta_t$ denote its inequality slack at hash share $t$, and let
$\Gamma_t$ denote the corresponding auxiliary slack obtained by
separating the $t$-dependence of that centered one-step expression. Direct
substitution under the scaling above gives
\[
  \Delta_p
  =
  \theta\{\Delta_c+(c-p)\Gamma_c\}.
  \tag{3.33}\label{eq:3-33}
\]
Define the scaling slack
\[
  \Xi_t=\Delta_t+t\Gamma_t.
  \tag{3.34}\label{eq:3-34}
\]
The LP imposes
\[
  \Delta_c\ge0,
  \qquad
  \Xi_c\ge0
\]
for every one-step constraint family, where the prefix-scaling constraint is sometimes an
existing constraint such as \ref{constr:c6-1}, \ref{constr:c1-3}, or \ref{constr:c0-11}.  Hence
\[
\begin{aligned}
\Delta_p
&=
\theta
\left\{
  \frac{p}{c}\Delta_c
  +
  \frac{c-p}{c}\Xi_c
\right\}
\ge0.
\end{aligned}
\tag{3.35}\label{eq:3-35}
\]
The scaling slack itself is homogeneous:
\[
  \Xi_p=\theta\Xi_c\ge0.
  \tag{3.36}\label{eq:3-36}
\]

For affine-tail constraints, the same identities apply pointwise in the
tail parameter $b$. Therefore slopes and left-endpoint values are
preserved whenever the corresponding inequality slack and scaling slack are
constrained at $c$.

\paragraph{C2: Strict catch-up Wait constraints.}

The finite strict catch-up Wait constraints are exactly of the generic form
\[
  A=S_{a,b},
  \qquad
  Y=Y_{a,b},
  \qquad
  Z=Z_{a,b}.
\]
The prefix-scaling constraint $S_{a,b}\ge Z_{a,b}$ is exactly the finite
Hedge-monotonicity constraint \ref{constr:c6-1}. Thus \ref{constr:c2-1} is preserved by
\textup{(\ref{eq:3-35})}, and the scaling slack is already preserved by
\textup{(\ref{eq:3-36})} as part of C6.

For affine strict catch-up tails, the inequality slack
\[
  \Delta^W_d(b)
\]
is preserved by the same generic identity, using the affine
Hedge-monotonicity slack \ref{constr:c6-2} as the scaling slack.  Therefore the
slope and left-endpoint constraints \ref{constr:c2-3}--\ref{constr:c2-4} are preserved for every
$1\le d\le D$.

\paragraph{C3: Private diagonal constraints.}

The private diagonal Wait constraints are of the generic form
\[
  A=\mathcal D^\sigma(b),
  \qquad
  Y=L(b+1,b),
  \qquad
  Z=\mathcal V(b,b+1),
\]
with $\sigma\in\{-,+\}$.  Hence the finite constraints \ref{constr:c3-1}--\ref{constr:c3-3} and the tail constraints \ref{constr:c3-6}--\ref{constr:c3-11} are preserved.

The private diagonal Match constraints are homogeneous:
\[
  \mathcal D^-(b)\ge \mathcal E^-_0(b),
  \qquad
  \mathcal D^+(b)\ge \mathcal E^+_0(b).
\]
Their slacks are multiplied by $\theta$, so \ref{constr:c3-4}--\ref{constr:c3-5} and \ref{constr:c3-12}--\ref{constr:c3-15} are preserved.

\paragraph{C4: Public tie constraints.}

For reserve-zero public tie states, let
\[
  \Delta^{\sigma,p}_{E0,W}(b)
\]
be the reserve-zero TieWait slack at hash share $p$, and let
$\Xi^{\sigma,p}_{E0,W}(b)$ be the scaling slack from
\ref{constr:c4-2} or \ref{constr:c4-6}. These slacks satisfy the
prefix-scaling identity \textup{(\ref{eq:3-35})}, and the scaling slack
is homogeneous as in \textup{(\ref{eq:3-36})}. Hence
\ref{constr:c4-1}--\ref{constr:c4-8} are preserved.

For reserve-one public tie states, let
\[
  \Delta^{\sigma,p}_{E1,W}(b)
\]
be the reserve-one TieWait slack at hash share $p$, and let
$\Xi^{\sigma,p}_{E1,W}(b)$ be the scaling slack from
\ref{constr:c4-10} or \ref{constr:c4-14}. The same prefix-scaling
identity preserves \ref{constr:c4-9}--\ref{constr:c4-16}.

For reserve $r\ge2$, no additional finite LP variables are introduced. The TieWait inequalities are consequences of C1, and TieOverride is handled through the same lead/Override dominance argument used in the soundness proof. Since C1 is preserved, the reserve-$r\ge2$ public tie part remains valid.

\paragraph{C5: Hedge-success upper constraints.}

Every Hedge-success upper constraint has the form
\[
  V(H_d)\le q.
\]
After scaling,
\[
  V_p(H_d)=\theta V_c(H_d)\le \theta q_c=q_p.
\]
Thus C5 is preserved, including the finite, affine-tail, and terminal Hedge successor cases.

\paragraph{C6: Hedge monotonicity constraints.}

Hedge monotonicity constraints are homogeneous:
\[
  V(a,b)\ge V(a,b+1).
\]
After scaling,
\[
  V_p(a,b)-V_p(a,b+1)
  =
  \theta\{V_c(a,b)-V_c(a,b+1)\}\ge0.
\]
Therefore \ref{constr:c6-1}--\ref{constr:c6-4} are preserved.  The terminal monotonicity part is already covered by \ref{constr:c0-11} and Lemma~\ref{lem:terminal-extension}.

\paragraph{C7: One-block lead Wait constraints.}

The one-block lead Wait constraints are of the generic form
\[
  A=L(b+1,b),
  \qquad
  Y=L(b+2,b),
  \qquad
  Z=\mathcal D^+(b+1).
\]
The prefix-scaling constraint is exactly \ref{constr:c7-2} in the finite region and
\ref{constr:c7-6}--\ref{constr:c7-8} in the tail region. Hence the finite constraints
\ref{constr:c7-1}--\ref{constr:c7-2} and the tail constraints \ref{constr:c7-3}--\ref{constr:c7-8} are preserved by the
generic scaling calculation.

\paragraph{C8: Terminal strict catch-up constraints.}

The terminal boundary constraint at deficit $D+1$ is also of the generic form:
write $a=b-(D+1)$ and set
\[
\begin{aligned}
A&=aw_1,\\
Y&=C_D(b),\\
Z&=aw_2.
\end{aligned}
\]
The scaling slack is $aw_1\ge aw_2$, which follows from terminal
monotonicity \ref{constr:c0-11}. Therefore \ref{constr:c8-1}--\ref{constr:c8-3} are preserved by
\textup{(\ref{eq:3-35})}, and no additional C8 prefix-scaling constraint is needed.

For the explicit terminal superharmonic constraints, compute
\[
\begin{aligned}
& w_k^p-pw_{k-1}^p-q_pw_{k+1}^p \\
&\quad =
\theta(w_k^c-cw_{k-1}^c-q_cw_{k+1}^c)\\
&\qquad
+(c-p)\theta(w_{k-1}^c-w_{k+1}^c).
\end{aligned}
\tag{3.37}\label{eq:3-37}
\]
Here the scaling slack is
\[
  (w_k^c-cw_{k-1}^c-q_cw_{k+1}^c)
  +c(w_{k-1}^c-w_{k+1}^c)
  =
  w_k^c-w_{k+1}^c,
\]
which is nonnegative by terminal monotonicity \ref{constr:c0-11}. Hence \ref{constr:c8-4} is
preserved.  The terminal extension lower condition \ref{constr:c8-5} was already
checked under C0.

Thus every constraint of
\[
  \mathcal L_{N,D}(p,\gamma^-,\gamma^+)
\]
is satisfied by the scaled variables.  Therefore the LP is feasible at every $p<c$.
\end{proof}

%% file: sections/ideal-to-real.tex
\section{From IDEAL to REAL}
\label{sec:ideal-to-real}

We transfer the certificate from IDEAL to REAL in two steps. First, we lift
it to core-REAL by canonicalizing arbitrary core deviations into IDEAL
policies. Second, split accounting transfers the core-REAL bound to arbitrary
REAL deviations.

\subsection{From IDEAL to core-REAL}
\label{subsec:ideal-to-core-real}

Fix a tested miner $m_i$ and a strategic core $c_i$ of hash share $p$.
All operational identities other than $c_i$, including the
Frontier component of $m_i$ if present, follow the Frontier strategy.
Throughout this section, write
\[
  q=1-p.
\]
We aggregate these Frontier-following operational identities into a
Frontier aggregate of hash share $q$.  A Frontier discovery is always
understood together with the Frontier strategy's immediate publication
of the discovered block. Thus one effective Frontier transition includes
the newly published Frontier block.

A REAL history has a \emph{canonical prefix} if every action already
taken by $c_i$ is one of the canonical actions defined in
Section~\ref{sec:model-ideal}:
\[
  \begin{gathered}
    \mathsf{Wait},\quad
    \mathsf{Hedge},\quad
    \mathsf{NoMining},\quad
    \mathsf{NoRelease},\quad
    \mathsf{Override},\\
    \mathsf{Match},\quad
    \mathsf{TieWait},\quad
    \mathsf{TieOverride}.
  \end{gathered}
\]

By Lemma~\ref{lem:canonical-state-representation}, every canonical-prefix
REAL history has a unique current reduced state in the canonical state space
$\widehat{\mathcal X}$ defined in
Section~\ref{subsec:one-deviation-reduction}.

Fix a canonical REAL state $x$, and write $A(x)=A_{c_i}(x)$ for the
strategic core's feasible action set. For a continuation policy $\pi$
starting from $x$, let $r_{c_i,t}^{\pi}$ be the reward from blocks with
source component $c_i$ that become newly included in the common trunk in
transition $t$, and let $\ell_t^{\pi}$ be its public-frontier
advancement. Define
\[
  R_{c_i,n}^{\pi}
  :=\sum_{t=0}^{n-1}r_{c_i,t}^{\pi},
  \qquad
  L_n^{\pi}
  :=\sum_{t=0}^{n-1}\ell_t^{\pi}.
\]

Let $u,u'\in A(x)$. We say that $u'$ weakly dominates $u$ at $x$,
and write
\[
  u\preceq_x u',
\]
if the following holds. For every continuation policy $\pi$ starting
from $x$ whose first action is $u$, there exist a continuation policy
$\pi'$ starting from $x$ whose first action is $u'$ and a coupling of the
two continuations such that, on every coupled execution, for every finite
transition depth $n\in\mathbb Z_{\ge0}$,
\[
  L_n^{\pi'}=L_n^{\pi},
  \qquad
  R_{c_i,n}^{\pi'}\ge R_{c_i,n}^{\pi}.
\]

\begin{lemma}
\label{lem:real-local-dominance}
Let $x$ be a canonical REAL state. The following action dominances hold
at $x$.

\begin{enumerate}
  \item First, suppose that a release action $u$ publishes one or more
  blocks on a non-active private branch. Let $u'$ be the release action
  obtained by omitting all non-active blocks while leaving every other
  released block unchanged. Then
  \[
    u\preceq_x u'.
  \]

  Second, suppose that a release action $u$ publishes only blocks on the
  active core-side branch but creates neither a public tie nor a strictly
  longest core-side branch. Let $\bar u=\mathsf{NoRelease}$. Then
  \[
    u\preceq_x \bar u.
  \]

  \item Suppose there is an active core-side branch $B$ of height
  $H\ge1$. Let $A$ be a mining action that does not extend the tip of
  $B$, and suppose an $m_i$-success under $A$ creates a branch of
  unresolved height $\ell\le H+1$. Let $N=\mathsf{NoMining}$, and let
  $W$ be the canonical action that extends the tip of $B$, namely
  $\mathsf{Wait}$ at a non-public-tie state and $\mathsf{TieWait}$ at a
  public-tie state. Then
  \[
    A\preceq_x N
    \quad\text{if }\ell\le H,
    \qquad
    A\preceq_x W
    \quad\text{if }\ell=H+1.
  \]

\end{enumerate}
\end{lemma}

\begin{proof}
For part~1, fix a continuation policy $\pi$ whose first action is $u$.
In the first case, replace $u$ by $u'$; any active-branch publication,
including one that creates a tie or an override, is retained. The omitted
branches are strictly below the relevant public branch by the canonical-state
definition and therefore do not affect the public longest-branch comparison.
In the second case, replace $u$ by $\bar u$; this also preserves that
comparison because $u$ creates neither a tie nor an override. In both cases,
keep the omitted blocks private and simulate $\pi$. Couple all later actions,
discoveries, and tie-breaking outcomes identically. This coupling witnesses
the required dominance and proves part~1.

For part~2, fix a continuation policy $\pi$ whose first action is $A$, and
let $B_A$ be the branch created by an $m_i$-success. First suppose
$\ell\le H$ and replace $A$ by $N$. Couple the state-preserving
probability-$p$ branch of $N$ with the $A$-success branch, and its
probability-$q$ state-changing branch with the Frontier-success branch under
$A$. After the state-preserving branch, let the next action begin a simulation
of the continuation after the $A$-success, replacing $B_A$ by the
height-$\ell$ prefix of $B$. Map every later
mining or release action involving $B_A$ to this substituted prefix and
copy all other actions and discoveries. At every coupled increase of the
maximum public height, the two continuations compare the same private and
public heights and therefore update to the same tie-origin label. Together
with the fixed total Frontier allocation in a multi-branch tie, this
preserves every later tie resolution. The paired first transitions have zero
reward and advancement, and all later transitions are paired one-for-one.
Thus public evolution and advancement are identical.
If $B_A$ enters the common trunk, its substitute contains weakly more
$m_i$-blocks; otherwise the reward is unchanged. Hence $A\preceq_x N$.

Now suppose $\ell=H+1$ and replace $A$ by $W$. Couple the next discovery
and any tie-breaking outcome. A Frontier discovery gives the same successor.
On an $m_i$-success, $W$ creates an all-$m_i$ path $B^+$ of height $H+1$,
the same height as $B_A$. Simulate the continuation after the $A$-success by
replacing $B_A$ with $B^+$, mapping later mining and release actions to the
substituted path and copying all other actions and discoveries. Again, the
public evolution and advancement are identical, while every common-trunk
prefix contains weakly more $m_i$-blocks. Hence $A\preceq_x W$.

\end{proof}

We define a \emph{policy tree rooted at $x$} as a rooted tree whose root
state is $x$. Each nonterminal node is a state--action pair $(y,u)$, where $u$ is
feasible at $y$, and each possible transition under $u$ gives an edge to
a child whose state is the corresponding successor. Different nodes may
have the same state and choose different actions. A policy induces a policy
tree by assigning to each node the action prescribed for its root-to-node
history; conversely, a policy tree specifies a policy on its histories.
A policy tree rooted at the initial state is called simply a \emph{policy
tree}.

\begin{theorem}
\label{thm:real-to-ideal-policy}
For every feasible core policy $\pi^c$ of the strategic core $c_i$, there
exists an IDEAL policy $\widehat\pi$, using only the
release actions
\[
  \mathrm{NoRelease},\quad
  \mathrm{Override},\quad
  \mathrm{Match},\quad
  \mathrm{TieOverride},
\]
and the mining actions
\[
  \mathrm{Wait},\quad
  \mathrm{Hedge},\quad
  \mathrm{NoMining},\quad
  \mathrm{TieWait},
\]
such that
\[
  U^{\mathrm{core}}_p(\pi^c;\gamma^-,\gamma^+)
  \le
  U^{\mathrm{core}}_p(\widehat\pi;\gamma^-,\gamma^+).
\]
\end{theorem}

\begin{proof}
Let $T_0$ be the policy tree induced by $\pi^c$. Construct a sequence
\[
  T_0,T_1,T_2,\ldots
\]
by breadth-first canonicalization. At stage $j$, a processed
prefix contains only canonical actions. At the first node following such a
prefix whose action is noncanonical, apply
Lemma~\ref{lem:real-local-dominance} to replace it by a canonical action.
Use the dominating residual continuation supplied by the lemma to obtain
$T_{j+1}$ from $T_j$. The
breadth-first construction replaces every node only finitely many times and
therefore defines a unique limit tree $\widehat T$.

Every action in $\widehat T$ is canonical. Fix $j\ge0$. The remaining breadth-first replacements
from $T_j$ to $\widehat T$ are weakly dominating, and their couplings compose.
Thus, on every coupled execution and at every finite transition depth $n$,
\[
  L_n^{\widehat T}=L_n^{T_j},
  \qquad
  R_{c_i,n}^{\widehat T}\ge R_{c_i,n}^{T_j}.
\]
Taking expectations and the long-run liminf gives, for every $j\ge0$,
\[
  U^{\mathrm{core}}_p(T_j;\gamma^-,\gamma^+)
  \le
  U^{\mathrm{core}}_p(\widehat T;\gamma^-,\gamma^+).
\]
Hence the utility of $\widehat T$ is an upper bound on the utilities of all
$T_j$, including $T_0$.

Each replacement does not increase the number of release actions in a release
phase, so $\widehat T$ satisfies the restriction on consecutive
$\mathsf{NoRelease}$ actions. Thus
$\widehat T$ specifies an IDEAL policy $\widehat\pi$. Since $T_0$ is the
policy tree of $\pi^c$, it follows that
\[
  U^{\mathrm{core}}_p(\pi^c;\gamma^-,\gamma^+)
  \le
  U^{\mathrm{core}}_p(\widehat\pi;\gamma^-,\gamma^+).
\]
\end{proof}

\begin{theorem}
\label{thm:ideal-to-core-real}
Suppose the certificate LP
\[
  \mathcal L_{N,D}(p,\gamma^-,\gamma^+)
\]
is feasible. Then every feasible core policy $\pi^c$ of a strategic
core with hash share $p$ satisfies
\[
  U^{\mathrm{core}}_p(\pi^c;\gamma^-,\gamma^+)\le p,
\]
where $U^{\mathrm{core}}_p$ is the long-run ratio of cumulative
strategic-core reward to cumulative public-frontier advancement.
\end{theorem}

\begin{proof}
Fix any core policy $\pi^c$. By
Theorem~\ref{thm:real-to-ideal-policy}, there is an IDEAL
policy $\widehat\pi$ such that
\[
  U^{\mathrm{core}}_p(\pi^c;\gamma^-,\gamma^+)
  \le
  U^{\mathrm{core}}_p(\widehat\pi;\gamma^-,\gamma^+).
\]
Feasibility of $\mathcal L_{N,D}(p,\gamma^-,\gamma^+)$ gives a potential
satisfying every centered Bellman inequality.  By
Lemma~\ref{lem:ideal-policy-bound},
\[
  U^{\mathrm{core}}_p(\widehat\pi;\gamma^-,\gamma^+)\le p.
\]
Therefore
\[
  U^{\mathrm{core}}_p(\pi^c;\gamma^-,\gamma^+)
  \le
  U^{\mathrm{core}}_p(\widehat\pi;\gamma^-,\gamma^+)
  \le p.
\]
\end{proof}

\subsection{From core-REAL to REAL}
\label{subsec:core-real-to-real}

\begin{theorem}
\label{thm:lp-to-real}
Let $h^{\mathrm{LB}}<1/2$ satisfy
\[
  \mathcal L_{N,D}
  (h^{\mathrm{LB}},\gamma^-,\gamma^+)
  \quad\text{is feasible}.
\]
If every miner $m_i\in\mathcal M$ satisfies
\[
  \alpha_i\le h^{\mathrm{LB}},
\]
then the Frontier profile is a Nash equilibrium in the REAL model.
\end{theorem}

\begin{proof}
By prefix monotonicity, feasibility at $h^{\mathrm{LB}}$ implies
feasibility at every $\beta\le h^{\mathrm{LB}}$. By
Theorem~\ref{thm:ideal-to-core-real}, every core policy of core
share $\beta\le h^{\mathrm{LB}}$ satisfies
\[
  U^{\mathrm{core}}_\beta(\pi^c)\le \beta.
\]

Fix a miner $m_i$ with $\alpha_i\le h^{\mathrm{LB}}$ and a REAL
deviation $(p,\pi^c)$, where $0\le p\le1$. Set
\[
  \beta=\alpha_i p,
  \qquad
  k=\frac{\alpha_i(1-p)}{1-\alpha_i p}.
\]
Then $\beta\le h^{\mathrm{LB}}$, so
\[
  U^{\mathrm{core}}_\beta(\pi^c)\le\beta.
\]
Moreover, $0\le k\le1$ because $\alpha_i<1$.

At transition depth $n$, let $C_n$ be the cumulative common-trunk reward
whose source component is the strategic core, and let $F_n$ be the
cumulative common-trunk reward whose source component belongs to the
Frontier aggregate. Since Frontier components are operationally
indistinguishable, the Frontier component of $m_i$ receives the fraction
$k$ of the expected Frontier-aggregate reward. Hence
\[
  \mathbb E[\mathcal R_{i,n}]
  =
  \mathbb E[C_n]+k\mathbb E[F_n].
\]
The common trunk cannot extend beyond the maximum public height. Therefore,
for the cumulative public-frontier advancement $\mathcal L_n$,
\[
  C_n+F_n\le\mathcal L_n
\]
pathwise. It follows that
\[
\begin{aligned}
  \mathbb E[\mathcal R_{i,n}]
  &=k\mathbb E[C_n+F_n]+(1-k)\mathbb E[C_n] \\
  &\le k\mathbb E[\mathcal L_n]+(1-k)\mathbb E[C_n].
\end{aligned}
\]
Whenever $\mathbb E[\mathcal L_n]>0$, division by
$\mathbb E[\mathcal L_n]$ gives
\[
  \frac{\mathbb E[\mathcal R_{i,n}]}{\mathbb E[\mathcal L_n]}
  \le
  k+(1-k)\frac{\mathbb E[C_n]}{\mathbb E[\mathcal L_n]}.
\]
Taking the long-run lower limit and using the core bound yields
\[
\begin{aligned}
  U_i((p,\pi^c),\Sigma^{\mathsf F}_{-i})
  &\le k+(1-k)U^{\mathrm{core}}_\beta(\pi^c) \\
  &\le k+(1-k)\beta \\
  &=\alpha_i.
\end{aligned}
\]
Thus no REAL deviation gives $m_i$ utility greater than $\alpha_i$.
Since $m_i$ was arbitrary, Frontier is a Nash equilibrium.
\end{proof}

%% file: sections/algorithm.tex
\section{Algorithm for Lower Bounds}
\label{sec:algorithm}

We now give an algorithm for computing a lower bound for the REAL model. It
searches for a feasible certificate LP in the IDEAL
model. Once such a certificate is found,
Theorems~\ref{thm:ideal-to-core-real} and~\ref{thm:lp-to-real} give the
corresponding REAL guarantee.

The input is a REAL-model parameter tuple $(\gamma^-,\gamma^+)$, LP
parameters $(N,D)$, and a precision parameter $\varepsilon$. The output is
the certified lower bound
\[
  \alpha^{\mathrm{LB}}_{N,D}(\gamma^-,\gamma^+)
\]
such that, if every miner $m_i\in\mathcal M$ satisfies
\[
  \alpha_i\le \alpha^{\mathrm{LB}}_{N,D}(\gamma^-,\gamma^+),
\]
then the Frontier profile is a Nash equilibrium in the REAL model.

Let
\[
  \mathcal L_{N,D}(p,\gamma^-,\gamma^+)
\]
denote the certificate LP defined in Section~\ref{sec:ideal-certificate-lp}.

\begin{algorithm}[t]
\caption{REAL lower-bound algorithm}
\label{alg:real-lower-bound}
\begin{algorithmic}[1]
\REQUIRE $\gamma^-,\gamma^+,N,D,\varepsilon$
\ENSURE $\alpha^{\mathrm{LB}}_{N,D}(\gamma^-,\gamma^+)$
\STATE $\ell \leftarrow 0$
\STATE $u \leftarrow 1/2$
\STATE $\alpha^{\mathrm{LB}}_{N,D}(\gamma^-,\gamma^+)\leftarrow 0$
\WHILE{$u-\ell>\varepsilon$}
  \STATE $p\leftarrow(\ell+u)/2$
  \IF{$\mathcal L_{N,D}(p,\gamma^-,\gamma^+)$ is feasible}
    \STATE $\ell\leftarrow p$
    \STATE $\alpha^{\mathrm{LB}}_{N,D}(\gamma^-,\gamma^+)\leftarrow p$
  \ELSE
    \STATE $u\leftarrow p$
  \ENDIF
\ENDWHILE
\STATE \textbf{return} $\alpha^{\mathrm{LB}}_{N,D}(\gamma^-,\gamma^+)$
\end{algorithmic}
\end{algorithm}

\begin{theorem}
\label{thm:algorithm-soundness}
Suppose Algorithm~\ref{alg:real-lower-bound} returns
\[
  \alpha^{\mathrm{LB}}_{N,D}(\gamma^-,\gamma^+).
\]
Then the returned value is a lower bound on the REAL
incentive-compatibility threshold:
\[
  \alpha^{\mathrm{LB}}_{N,D}(\gamma^-,\gamma^+)
  \le
  \alpha^*(\gamma^-,\gamma^+).
\]
\end{theorem}

\begin{proof}
By Theorem~\ref{thm:prefix-monotonicity}, LP feasibility is prefix-monotone
in $p$, so Algorithm~\ref{alg:real-lower-bound} returns a certified feasible
point. Theorem~\ref{thm:lp-to-real} then places this point in
$\mathcal I^{\mathrm{REAL}}(\gamma^-,\gamma^+)$, proving the claim.
\end{proof}

%% file: sections/closed-form-upper-bound.tex
\section{Algorithm for Upper Bounds}
\label{sec:closed-form-upper-bound}

This section gives an algorithm for computing an upper bound on
the incentive-compatibility threshold from fixed tie-breaking parameters
$(\gamma^-,\gamma^+)$. We first define two explicit trigger policies and
use SM1 as a third policy~\cite{eyal2014majority}. We then derive the
trigger-policy gains in
closed form and minimize the hash shares at which the candidate policies
become profitable.

\subsection{Candidate Deviation Policies}

Fix a tested miner with hash share $p\in(0,1/2)$ and write $q=1-p$. For an
absolute Frontier-side height $b$, abbreviate
\[
  T_1(b):=(b-1,b),
  \qquad
  T_2(b):=(b-2,b).
\]
Thus $T_1(b)$ and $T_2(b)$ are states in which the tested miner trails by one
and two blocks, respectively.

The plus-trigger and minus-trigger algorithms use the same matched-trail
rule with an integer trigger height $B\ge3$:
\begin{enumerate}
\item At $D_b^\sigma$, apply Match and enter $E_{b,0}^\sigma$.
\item At $E_{b,0}^\sigma$, apply TieWait. A tested-miner discovery enters
$E_{b,1}^\sigma$ and is followed by TieOverride$(1)$, which settles $b+1$
tested-miner blocks. A Frontier discovery on the core-side tied branch
settles $b$ tested-miner blocks and one Frontier block. A Frontier discovery
on the competing branch enters $T_1(b+1)$.
\item At $T_1(b)$, apply Hedge$(1)$ if $b=1$ and Wait if $b\ge2$. A
tested-miner discovery enters $D_b^-$, and a Frontier discovery enters
$T_2(b+1)$.
\item At $T_2(h)$, end the current cycle and begin a new cycle under Frontier
mining if $h<B$, and apply Wait if
$h\ge B$. In the latter case, a tested-miner discovery returns to $T_1(h)$,
whereas a Frontier discovery ends the current cycle and begins a new cycle
under Frontier mining.
\end{enumerate}

\begin{inlinealgorithm}{Plus-trigger policy $\pi_B^+$}
\label{alg:plus-trigger}
\begin{algorithmic}[1]
\REQUIRE Trigger height $B\ge3$
\STATE Start a renewal cycle from the initial public state
\IF{Frontier discovers the first block}
  \STATE End the current cycle and begin a new cycle under Frontier mining
\ELSE
  \STATE Withhold the first tested-miner block
\ENDIF
\IF{the tested miner discovers the second block}
  \STATE Mine privately until the lead falls from two to one
  \STATE Publish the private branch, override, and end the cycle
\ELSE
  \STATE Enter $D_1^+$ and follow the matched-trail rule with trigger $B$
\ENDIF
\end{algorithmic}
\end{inlinealgorithm}

\begin{inlinealgorithm}{Minus-trigger policy $\pi_B^-$}
\label{alg:minus-trigger}
\begin{algorithmic}[1]
\REQUIRE Trigger height $B\ge3$
\STATE Start a renewal cycle from the initial public state
\IF{the tested miner discovers the first block}
  \STATE Release that block and end the cycle
\ELSE
  \STATE Enter $T_1(1)$ and follow the matched-trail rule with trigger $B$
\ENDIF
\end{algorithmic}
\end{inlinealgorithm}

The two trigger policies differ only in how they enter the matched-trail
subsystem. The plus-trigger policy starts a private branch after the tested
miner's first discovery, whereas the minus-trigger policy starts one block behind
after a first discovery by Frontier. We also use SM1, the standard one-block
selfish-mining policy, which uses only the plus-origin tie-breaking parameter.

\subsection{Preliminaries for the Gain Analysis}

For an explicit regenerative policy $\pi$, consider one renewal cycle. Let
$R_\pi$ be the reward credited to the tested miner for blocks newly included
in the common trunk during the cycle, and let $L_\pi$ be the
public-frontier advancement during the cycle. If the cycle has finite expected
length and $\mathbb E[L_\pi]>0$, then
\[
  U_\pi
  =
  \frac{\mathbb E[R_\pi]}{\mathbb E[L_\pi]}.
\]
Define the centered renewal gain
\[
  \Delta_\pi(p)
  :=
  \mathbb E[R_\pi-pL_\pi].
  \tag{6.1}\label{eq:renewal-centered-gain}
\]
Then
\[
  U_\pi-p
  =
  \frac{\Delta_\pi(p)}{\mathbb E[L_\pi]}.
  \tag{6.2}\label{eq:renewal-utility-gap}
\]
Thus every $p$ at which an explicit policy has positive centered gain is an
upper bound on the incentive-compatibility threshold.

We next derive the recurrence shared by the plus-trigger and minus-trigger
families. Let $F_b$ be the expected centered reward from $T_1(b)$, let $J_b$
be the value from $T_2(b)$, and let $M_\sigma(b)$ be the value after Match at
$D_b^\sigma$. Resolving the public tie gives
\[
\begin{aligned}
  M_\sigma(b)
  ={}&pq(b+1)+\gamma^\sigma q(qb-p)\\
     &+(1-\gamma^\sigma)qF_{b+1}.
\end{aligned}
\tag{6.3}\label{eq:attack-tie-value}
\]
For $b\ge B$, the policy continues from a two-block deficit, so
\[
  F_b=pM_-(b)+qJ_{b+1},
  \qquad
  J_b=pF_b-pq(b+1).
  \tag{6.4}\label{eq:attack-affine-recurrence}
\]
The solution is affine in $b$. Define
\[
\begin{gathered}
  r_3:=pq(2-\gamma^-),
  \qquad s:=1-r_3,\\
  A:=\frac{pq(p-q+\gamma^-q)}{s},\\
  C:=\frac{r_3A+pq(p-2q-\gamma^-p)}{s}.
\end{gathered}
\tag{6.5}\label{eq:attack-affine-coefficients}
\]
Then
\[
  F_{B-1}=A(B-1)+C.
  \tag{6.6}\label{eq:attack-affine-boundary}
\]

Below the trigger height, a Frontier discovery from a two-block deficit ends
the cycle. With
\[
  r_2:=pq(1-\gamma^-),
\]
the finite-prefix recurrence is
\[
  F_b
  =
  r_2F_{b+1}
  -pq\bigl[q(1-\gamma^-)b+q+\gamma^-p\bigr].
  \tag{6.7}\label{eq:attack-prefix-recurrence}
\]
Hence
\[
\begin{aligned}
  F_2
  ={}&r_2^{B-3}\{A(B-1)+C\}\\
  &-pq\sum_{j=0}^{B-4}r_2^j
  \bigl[q(1-\gamma^-)(j+2)+q+\gamma^-p\bigr],
\end{aligned}
\tag{6.8}\label{eq:attack-f2}
\]
and
\[
\begin{aligned}
  F_1
  ={}&r_2^{B-2}\{A(B-1)+C\}\\
  &-pq\sum_{j=0}^{B-3}r_2^j
  \bigl[q(1-\gamma^-)(j+1)+q+\gamma^-p\bigr].
\end{aligned}
\tag{6.9}\label{eq:attack-f1}
\]
An empty sum is interpreted as zero. Since $p<q$ and
$pq(2-\gamma^-)\le1/2$, all three policies have finite expected
cycle length.

\subsection{Plus-Trigger Gain}

For Algorithm~\ref{alg:plus-trigger}, first condition on the tested miner
discovering the first two blocks. The private lead is then two. Let $\tau$ be
the number of subsequent discoveries until this lead first falls to one, and
let $A_\tau$ and $H_\tau$ be the numbers of tested-miner and Frontier
discoveries before that time. At stopping,
\[
  H_\tau-A_\tau=1.
\]
Since $p<q$, Wald's identity gives
\[
  (q-p)\mathbb E[\tau]=1,
  \qquad
  \mathbb E[H_\tau]=q\mathbb E[\tau]=\frac{q}{q-p}.
\]
When the lead reaches one, the policy overrides with a branch consisting
entirely of tested-miner blocks. The branch has length $H_\tau+1$. Its reward
and public-frontier advancement are both $H_\tau+1$, so the conditional
expected centered gain is
\[
  K_2
  =
  q\mathbb E[H_\tau+1]
  =
  q\left(1+\frac{q}{q-p}\right).
\]

If instead the tested miner discovers the first block and Frontier discovers
the second, the process reaches $D_1^+$. Substituting $b=1$ and $\sigma=+$
into Equation~\textup{(\ref{eq:attack-tie-value})} yields
\[
  M_+(1)
  =
  2pq+\gamma^+q(q-p)+(1-\gamma^+)qF_2.
\]
Finally, a first discovery by Frontier contributes $-p$ and ends the cycle.
Conditioning on the first two discoveries therefore gives
\[
  \Delta_B^+(p,\gamma^-,\gamma^+)
  =-pq+p\bigl[pK_2+qM_+(1)\bigr].
  \tag{6.10}\label{eq:plus-trigger-gain}
\]
The plus-trigger policy is profitable if and only if this gain is positive.

\subsection{Minus-Trigger Gain}

For Algorithm~\ref{alg:minus-trigger}, the centered renewal gain is
\[
  \Delta_B^-(p,\gamma^-)
  =
  pq+qF_1.
  \tag{6.11}\label{eq:minus-trigger-gain}
\]
The minus-trigger policy is profitable if and only if this gain is positive.
It does not depend on $\gamma^+$.

\subsection{Upper-Bound Algorithm}

For SM1~\cite{eyal2014majority}, the sign-change point is
\[
  \theta_{\mathrm{SM1}}(\gamma^+)
  =
  \frac{1-\gamma^+}{3-2\gamma^+}.
  \tag{6.12}\label{eq:sm1-threshold}
\]

For each $B\ge3$, define
\[
\begin{aligned}
  \theta_B^+(\gamma^-,\gamma^+)
  &:=
  \inf\{p\in(0,1/2):\Delta_B^+(p,\gamma^-,\gamma^+)>0\},\\
  \theta_B^-(\gamma^-)
  &:=
  \inf\{p\in(0,1/2):\Delta_B^-(p,\gamma^-)>0\}.
\end{aligned}
\tag{6.13}\label{eq:trigger-thresholds}
\]
The infimum of an empty set is $+\infty$.

\begin{inlinealgorithm}{REAL upper-bound algorithm}
\label{alg:attack-upper}
\begin{algorithmic}[1]
\REQUIRE $\gamma^-,\gamma^+$ and an integer $B_{\max}\ge3$
\ENSURE $\alpha^{\mathrm{UB}}_{B_{\max}}(\gamma^-,\gamma^+)$
\STATE $u\leftarrow\theta_{\mathrm{SM1}}(\gamma^+)$
\FOR{$B=3,\ldots,B_{\max}$}
  \STATE Compute $\theta_B^+$ and $\theta_B^-$ from the signs of
  $\Delta_B^+$ and $\Delta_B^-$
  \STATE $u\leftarrow\min\{u,\theta_B^+,\theta_B^-\}$
\ENDFOR
\STATE \textbf{return} $\alpha^{\mathrm{UB}}_{B_{\max}}(\gamma^-,\gamma^+)\leftarrow u$
\end{algorithmic}
\end{inlinealgorithm}

\begin{theorem}
\label{thm:attack-upper-soundness}
For every $B_{\max}\ge3$ and $(\gamma^-,\gamma^+)\in[0,1]^2$,
Algorithm~\ref{alg:attack-upper} returns an upper bound on the REAL
incentive-compatibility threshold:
\[
  \alpha^*(\gamma^-,\gamma^+)
  \le
  \alpha^{\mathrm{UB}}_{B_{\max}}(\gamma^-,\gamma^+).
\]
\end{theorem}

\begin{proof}
Each candidate in Algorithm~\ref{alg:attack-upper} is the infimum of hash
shares at which one explicit REAL policy has positive centered gain. By
Equation~\textup{(\ref{eq:renewal-utility-gap})}, every such positive-gain
point admits a profitable unilateral deviation. The threshold is no
larger than the infimum of those points for each policy, and hence no larger
than their minimum.
\end{proof}

%% file: sections/gamma-monotonicity.tex
\section{\texorpdfstring{$\gamma$}{Gamma}-Monotonicity}
\label{sec:gamma-monotonicity}

The REAL incentive-compatibility threshold is coordinatewise nonincreasing in
$(\gamma^-,\gamma^+)$. Write
\[
  \Gamma=(\gamma^-,\gamma^+)
\]
and abbreviate
\[
  \mathcal I^{\mathrm{REAL}}(\Gamma)
  :=\mathcal I^{\mathrm{REAL}}(\gamma^-,\gamma^+),
  \qquad
  \alpha^*(\Gamma)
  :=\alpha^*(\gamma^-,\gamma^+).
\]
We order tie-breaking pairs coordinatewise:
\[
  \Gamma_{\mathrm L}\preceq\Gamma_{\mathrm H}
  \quad\Longleftrightarrow\quad
  \gamma^-_{\mathrm L}\le\gamma^-_{\mathrm H}
  \ \text{and}\ 
  \gamma^+_{\mathrm L}\le\gamma^+_{\mathrm H}.
\]

Fix a strategic-core hash share $p\in(0,1)$ and set $q=1-p$.  For $n\ge0$,
let $J_n^\Gamma(x)$ be the supremum expected cumulative
centered gain over all canonical IDEAL policy trees rooted at $x$
and having depth at most $n$.  Stopping at any node is allowed.  Hence
\[
  J_0^\Gamma(x)=0
\]
For $n\ge1$, define
\[
\begin{aligned}
  Q_n^\Gamma(x,u)
  &:=
  \mathbb E_\Gamma\!\bigl[
    R(x,u,X')-pL(x,u,X')\\
  &\hspace{24mm}
    {}+J_{n-1}^\Gamma(X')\mid x,u
  \bigr].
\end{aligned}
\]
The Bellman recursion is
\[
  J_n^\Gamma(x)
  =
  \max\left\{
    0,
    \max_{u\in\widehat{\mathcal A}(x)}Q_n^\Gamma(x,u)
  \right\}.
  \tag{7.1}\label{eq:finite-depth-recursion}
\]
These finite policy trees are analytical objects; their nodes need not
merge when they have the same reduced state.

At a public tie state $E_{b,r}^\sigma$, the two Frontier-success
successors under $\mathsf{TieWait}$ are
\[
  \begin{aligned}
    C_{b,r}
    &:={}
    \begin{cases}
      (0,1), & r=0,\\
      D_1^+, & r=1,\\
      (r,1)^+, & r\ge2,
    \end{cases}\\
    F_{b,r}
    &:={}
    \begin{cases}
      (b,b+1), & r=0,\\
      D_{b+1}^+, & r=1,\\
      (b+r,b+1)^+, & r\ge2.
    \end{cases}
  \end{aligned}
\]
The state $C_{b,r}$ follows a discovery on the core-side tied branch and
immediately fixes its length-$b$ all-core prefix.  The state $F_{b,r}$
follows a discovery on the competing Frontier-side branch.

\begin{lemma}
\label{lem:tie-resolution-dominance}
For every $\Gamma$, $n\ge0$, $b\ge1$, and $r\ge0$,
\[
  b+J_n^\Gamma(C_{b,r})
  \ge
  J_n^\Gamma(F_{b,r}).
  \tag{7.2}\label{eq:tie-resolution-dominance}
\]
\end{lemma}

\begin{proof}
Take any depth-$n$ canonical policy tree rooted at $F_{b,r}$ and simulate
it from $C_{b,r}$. Apply every action to the corresponding suffix and
couple all discoveries and tie-breaking outcomes identically. The two
states have the same public height and relative lead $r-1$. They are both
unsigned strict catch-up states when $r=0$, plus-origin diagonal states when
$r=1$, and plus-origin lead states when $r\ge2$. Thus the simulation is
feasible and preserves both depth and public-frontier advancement.

The only difference is the length-$b$ prefix. At $C_{b,r}$, its $b$
all-$m_i$ blocks have already been fixed. If the simulated continuation
from $F_{b,r}$ later fixes the core-side prefix, this initial reward $b$
exactly matches its reward; if it fixes the competing Frontier prefix, the
$C_{b,r}$ execution earns $b$ more. Thus the simulated centered gain plus
$b$ is pathwise at least the original gain. Taking expectations and then
the supremum proves~\textup{(\ref{eq:tie-resolution-dominance})}.
\end{proof}

\begin{lemma}
\label{lem:finite-depth-gamma-monotonicity}
If $\Gamma_{\mathrm L}\preceq\Gamma_{\mathrm H}$, then for every canonical
state $x$ and every $n\ge0$,
\[
  J_n^{\Gamma_{\mathrm L}}(x)
  \le
  J_n^{\Gamma_{\mathrm H}}(x).
  \tag{7.3}\label{eq:finite-depth-gamma-order}
\]
\end{lemma}

\begin{proof}
We induct on $n$.  The claim is immediate for $n=0$.  Assume it holds at
depth $n-1$. Except for $\mathsf{TieWait}$ and $\mathsf{NoMining}$ at a
public tie state, the transition law and centered increment do not depend
directly on $\Gamma$. The induction hypothesis therefore orders the value of
every such action under $\Gamma_{\mathrm L}$ and $\Gamma_{\mathrm H}$.

It remains to consider $\mathsf{TieWait}$ and $\mathsf{NoMining}$. For
$\mathsf{TieWait}$ at $E_{b,r}^\sigma$, write $E'=E_{b,r+1}^\sigma$,
$C=C_{b,r}$, and $F=F_{b,r}$. Its depth-$n$ action value is
\[
\begin{aligned}
  Q_n^\Gamma
  ={}&pJ_{n-1}^\Gamma(E')
  +q\gamma^\sigma\{b-p+J_{n-1}^\Gamma(C)\}\\
  &+q(1-\gamma^\sigma)\{-p+J_{n-1}^\Gamma(F)\}.
\end{aligned}
\tag{7.4}\label{eq:tiewait-finite-value}
\]
Both Frontier-success outcomes advance the public frontier by one; the
core-side outcome additionally fixes $b$ core blocks. Their centered
increments are therefore $b-p$ and $-p$, respectively.
Let $\gamma_{\mathrm L}^\sigma$ and $\gamma_{\mathrm H}^\sigma$ denote the
corresponding coordinates of the two parameter pairs.  Expanding the
difference using the lower parameter as the baseline gives
\[
\begin{aligned}
  Q_n^{\Gamma_{\mathrm H}}-Q_n^{\Gamma_{\mathrm L}}
  ={}&p\{J_{n-1}^{\Gamma_{\mathrm H}}(E')
        -J_{n-1}^{\Gamma_{\mathrm L}}(E')\}\\
  &+q\gamma_{\mathrm L}^\sigma
    \{J_{n-1}^{\Gamma_{\mathrm H}}(C)
      -J_{n-1}^{\Gamma_{\mathrm L}}(C)\}\\
  &+q(1-\gamma_{\mathrm L}^\sigma)
    \{J_{n-1}^{\Gamma_{\mathrm H}}(F)
      -J_{n-1}^{\Gamma_{\mathrm L}}(F)\}\\
  &+q(\gamma_{\mathrm H}^\sigma-\gamma_{\mathrm L}^\sigma)
    \{b+J_{n-1}^{\Gamma_{\mathrm H}}(C)
      -J_{n-1}^{\Gamma_{\mathrm H}}(F)\}.
\end{aligned}
\tag{7.5}\label{eq:tiewait-gamma-difference}
\]
The first three terms are nonnegative by the induction hypothesis. The last
is nonnegative by Lemma~\ref{lem:tie-resolution-dominance}. Thus the
$\mathsf{TieWait}$ value is weakly larger under $\Gamma_{\mathrm H}$.

At $E=E_{b,r}^\sigma$, the depth-$n$ $\mathsf{NoMining}$ action value is
\[
\begin{aligned}
  Q_{n,\mathrm N}^\Gamma(E)
  ={}&pJ_{n-1}^\Gamma(E)
  +q\gamma^\sigma\{b-p+J_{n-1}^\Gamma(C)\}\\
  &+q(1-\gamma^\sigma)\{-p+J_{n-1}^\Gamma(F)\}.
\end{aligned}
\]
Using $\Gamma_{\mathrm L}$ as the baseline gives
\[
\begin{aligned}
&Q_{n,\mathrm N}^{\Gamma_{\mathrm H}}(E)
 -Q_{n,\mathrm N}^{\Gamma_{\mathrm L}}(E)\\
={}&p\{J_{n-1}^{\Gamma_{\mathrm H}}(E)
      -J_{n-1}^{\Gamma_{\mathrm L}}(E)\}\\
&+q\gamma_{\mathrm L}^\sigma
  \{J_{n-1}^{\Gamma_{\mathrm H}}(C)
    -J_{n-1}^{\Gamma_{\mathrm L}}(C)\}\\
&+q(1-\gamma_{\mathrm L}^\sigma)
  \{J_{n-1}^{\Gamma_{\mathrm H}}(F)
    -J_{n-1}^{\Gamma_{\mathrm L}}(F)\}\\
&+q(\gamma_{\mathrm H}^\sigma-\gamma_{\mathrm L}^\sigma)
  \{b+J_{n-1}^{\Gamma_{\mathrm H}}(C)
    -J_{n-1}^{\Gamma_{\mathrm H}}(F)\}.
\end{aligned}
\]
The induction hypothesis makes the first three terms nonnegative, and
Lemma~\ref{lem:tie-resolution-dominance} makes the last term nonnegative.
Thus every feasible action has weakly larger depth-$n$ value under
$\Gamma_{\mathrm H}$. Taking the maximum
in~\textup{(\ref{eq:finite-depth-recursion})} completes the induction.
\end{proof}

The lemma compares optimized finite-depth values, not the value of every
fixed policy.  After a favorable tie outcome becomes more likely, the
optimal continuation may change.

\begin{lemma}
\label{lem:finite-witness-implementation}
Fix $p\in(0,1)$ and $\Gamma$.  If
\[
  J_n^\Gamma(0,0)>0
\]
for some finite $n$, then there is a core-REAL policy
$\pi^c$ satisfying
\[
  U_p^{\mathrm{core}}(\pi^c;\Gamma)>p.
  \tag{7.6}\label{eq:finite-witness-profit}
\]
\end{lemma}

\begin{proof}
Choose a finite canonical policy tree $T$ with positive expected centered gain
\[
  G:=\mathbb E_\Gamma[R_T-pL_T]>0.
\]
The canonical tree can be executed directly in REAL.  At each leaf,
choose a public maximum tip, mine on it until the next discovery, and
release a core discovery immediately.  The resulting public longest branch
is unique and one block higher.  This reset has expected centered gain at
least
\[
  p(1-p)+(1-p)(-p)=0.
\]

Let $d_T\ge1$ be the depth of $T$ and set
\[
  M=2d_T.
\]
Every private or public branch created during $T$ has height at most $d_T$
above the cycle root. The reset adds one public block, and
$d_T+1\le M$. From the unique public tip, use Frontier-style public mining
until the cycle's total advancement is exactly $M$. Thus the final public
tip lies above every private and public branch created during $T$. Each
padding block has expected centered gain zero, so the resulting finite cycle
has advancement $M$ on every path and expected centered gain at least $G$.

Repeat this cycle above each successive unique public tip and ignore private
branches left by earlier cycles. This defines a core-REAL policy $\pi^c$.
Since every cycle has expected centered gain at least $G$ and advancement
$M$,
\[
  U_p^{\mathrm{core}}(\pi^c;\Gamma)
  \ge p+\frac{G}{M}>p.
\]
\end{proof}

\begin{theorem}
\label{thm:gamma-monotonicity}
If $\Gamma_{\mathrm L}\preceq\Gamma_{\mathrm H}$, then
\[
  \mathcal I^{\mathrm{REAL}}(\Gamma_{\mathrm H})
  \subseteq
  \mathcal I^{\mathrm{REAL}}(\Gamma_{\mathrm L}).
  \tag{7.7}\label{eq:gamma-ic-inclusion}
\]
Consequently,
\[
  \alpha^*(\Gamma_{\mathrm L})
  \ge
  \alpha^*(\Gamma_{\mathrm H}).
  \tag{7.8}\label{eq:gamma-threshold-order}
\]
\end{theorem}

\begin{proof}
We prove the set inclusion by contraposition.  Suppose
\[
  c\notin\mathcal I^{\mathrm{REAL}}(\Gamma_{\mathrm L}).
\]
Then some REAL instance in which every miner has hash share at most $c$
admits a profitable deviation by a miner $m_i$ of share $\alpha_i\le c$.
Let $\beta$ be the operational hash share assigned to its core component.
The split-accounting argument in
Theorem~\ref{thm:lp-to-real} gives, for some $k\in[0,1]$,
\[
  U_i\le k+(1-k)U_\beta^{\mathrm{core}},
  \qquad
  \alpha_i=k+(1-k)\beta.
\]
Profitability therefore implies
\[
  U_\beta^{\mathrm{core}}>\beta,
  \qquad
  0<\beta\le c.
  \tag{7.9}\label{eq:profitable-core-extraction}
\]

Apply Theorem~\ref{thm:real-to-ideal-policy} under
$\Gamma_{\mathrm L}$. It yields an IDEAL policy with utility greater than
$\beta$. For some finite transition depth $n$, the
policy's expected cumulative centered gain is therefore positive. Unfolding
its first $n$ transitions gives a finite canonical policy tree, so
\[
  J_n^{\Gamma_{\mathrm L}}(0,0)>0.
\]
Lemma~\ref{lem:finite-depth-gamma-monotonicity} gives
\[
  J_n^{\Gamma_{\mathrm H}}(0,0)>0,
\]
and Lemma~\ref{lem:finite-witness-implementation} converts this witness
into a profitable core-REAL policy of share $\beta$ under
$\Gamma_{\mathrm H}$.

Finally, create a REAL instance with one strategic miner of share $\beta$
and split the remaining share $1-\beta$ equally among $m$ Frontier miners,
where
\[
  m\ge\left\lceil\frac{1-\beta}{c}\right\rceil.
\]
Every miner then has share at most $c$, but the strategic miner's core
policy is profitable.  Hence
\[
  c\notin\mathcal I^{\mathrm{REAL}}(\Gamma_{\mathrm H}).
\]
This proves~\textup{(\ref{eq:gamma-ic-inclusion})}.  Taking suprema gives
\textup{(\ref{eq:gamma-threshold-order})}.
\end{proof}

%% file: sections/arbitrary-tie-parameters.tex
\section{Algorithm for Arbitrary Tie Parameters}
\label{sec:arbitrary-tie-parameters}

In this section, we provide an algorithm that, for any tie-breaking pair,
returns lower and upper bounds on the
incentive-compatibility threshold together with a theoretical guarantee on
their accuracy.
The monotonicity result extends bounds computed on a finite grid to every
tie-breaking pair.  Fix a spacing $h=1/M$, where $M$ is a positive integer,
and consider the grid
\[
  \mathcal G_h=\{0,h,\ldots,1\}^2.
\]
Suppose that a certified lower bound $\alpha^{\mathrm{LB}}(\Gamma_g)$ and a
verified upper bound $\alpha^{\mathrm{UB}}(\Gamma_g)$ are available at the
required grid points.  For a target pair
$\Gamma=(\gamma^-,\gamma^+)\in[0,1]^2$, let
\[
\begin{gathered}
  i=\min\{\lfloor\gamma^-/h\rfloor,M-1\},\\
  j=\min\{\lfloor\gamma^+/h\rfloor,M-1\},\\
  \Gamma_L=(ih,jh),\qquad
  \Gamma_H=((i+1)h,(j+1)h).
\end{gathered}
\tag{8.1}\label{eq:grid-cell-corners}
\]
Here $\Gamma_L$ and $\Gamma_H$ denote the coordinatewise lower and upper
corners of the grid cell.  In coordinatewise order,
\[
  \Gamma_L\preceq\Gamma\preceq\Gamma_H.
\]

\begin{inlinealgorithm}{Threshold bounds at arbitrary tie parameters}
\label{alg:arbitrary-gamma-bounds}
\begin{algorithmic}[1]
\REQUIRE $\Gamma=(\gamma^-,\gamma^+)$, $h=1/M$, and grid bounds
$\alpha^{\mathrm{LB}},\alpha^{\mathrm{UB}}$
\ENSURE $[\underline\alpha_h(\Gamma),\overline\alpha_h(\Gamma)]$
\STATE Compute $\Gamma_L$ and $\Gamma_H$ from
Equation~\textup{(\ref{eq:grid-cell-corners})}
\STATE $\underline\alpha_h(\Gamma)\leftarrow
\alpha^{\mathrm{LB}}(\Gamma_H)$
\STATE $\overline\alpha_h(\Gamma)\leftarrow
\alpha^{\mathrm{UB}}(\Gamma_L)$
\STATE \textbf{return}
$[\underline\alpha_h(\Gamma),\overline\alpha_h(\Gamma)]$
\end{algorithmic}
\end{inlinealgorithm}

\begin{theorem}
\label{thm:arbitrary-gamma-bounds}
For every $\Gamma\in[0,1]^2$, Algorithm~\ref{alg:arbitrary-gamma-bounds}
returns bounds satisfying
\[
  \underline\alpha_h(\Gamma)
  \le
  \alpha^*(\Gamma)
  \le
  \overline\alpha_h(\Gamma).
  \tag{8.2}\label{eq:grid-envelope-guarantee}
\]
Moreover,
\[
  \overline\alpha_h(\Gamma)-\underline\alpha_h(\Gamma)
  \le
  \varepsilon_h,
\]
where
\[
  \varepsilon_h
  :=
  \max_{0\le i,j<M}
  \left\{
    \alpha^{\mathrm{UB}}(ih,jh)
    -\alpha^{\mathrm{LB}}((i+1)h,(j+1)h)
  \right\}.
  \tag{8.3}\label{eq:grid-envelope-width}
\]
\end{theorem}

\begin{proof}
Theorem~\ref{thm:gamma-monotonicity} and the pointwise guarantees of
Algorithms~\ref{alg:real-lower-bound} and~\ref{alg:attack-upper} give
\[
\begin{aligned}
  \alpha^{\mathrm{LB}}(\Gamma_H)
  &\le \alpha^*(\Gamma_H)
  \le \alpha^*(\Gamma)\\
  &\le \alpha^*(\Gamma_L)
  \le \alpha^{\mathrm{UB}}(\Gamma_L).
\end{aligned}
\]
The endpoints in this chain are exactly the values returned by the
algorithm.  For the cell indexed by $(i,j)$, their difference is
\[
  \alpha^{\mathrm{UB}}(ih,jh)
  -\alpha^{\mathrm{LB}}((i+1)h,(j+1)h).
\]
Taking the maximum over all cells proves the width bound.
\end{proof}

%% file: sections/near-tightness.tex
\section{Experiments}
\label{sec:near-tightness}

This section implements the lower- and upper-bound algorithms using exact
rational arithmetic and reports the maximum pointwise gap over the evaluated
parameter pairs and the uniform interval width guaranteed over $[0,1]^2$.

\subsection{Implementation}

For each $\gamma^-,\gamma^+\in\{0,0.001,\ldots,0.999\}$, we use exact
rational arithmetic to examine hash shares of the form $p=m/10^{10}$.
Algorithm~\ref{alg:real-lower-bound} with $N=D=20$ reports the largest
feasible value as the lower bound, while Algorithm~\ref{alg:attack-upper} with
$B_{\max}=24$ reports the smallest value with strictly positive centered gain
as the upper bound.

\subsection{Results}

Let
$\mathcal E:=\{0,0.001,\ldots,0.999\}^2$.
For each evaluated pair $\Gamma\in\mathcal E$, define the pointwise gap
\[
  d(\Gamma)
  :=
  \alpha^{\mathrm{UB}}(\Gamma)-\alpha^{\mathrm{LB}}(\Gamma).
\]
Exact rational evaluation gives
\[
  \max_{\Gamma\in\mathcal E} d(\Gamma)
  =
  \frac{2415}{10^{10}}
  =
  2.415\times10^{-7}
\]
For $\Gamma_{i,j}:=(0.001i,0.001j)$, define the width of grid cell
$(i,j)$ by
\[
  e_{i,j}
  :=
  \alpha^{\mathrm{UB}}(\Gamma_{i,j})
  -
  \alpha^{\mathrm{LB}}(\Gamma_{i+1,j+1}).
\]
The maximum grid-cell envelope width is
\[
  \varepsilon_{0.001}
  :=
  \max_{0\le i,j<1000} e_{i,j}
  =
  \frac{9{,}980{,}060}{10^{10}}
  =
  9.98006\times10^{-4}.
\]
Thus Algorithm~\ref{alg:arbitrary-gamma-bounds} returns, for every
$(\gamma^-,\gamma^+)\in[0,1]^2$, an interval containing
$\alpha^*(\gamma^-,\gamma^+)$ with width at most
$9.98006\times10^{-4}$.

%% file: sections/discussion.tex
\section{Discussion}
\label{sec:discussion}

\subsection{Multiple Ties}

The allocation of Frontier hash power among multiple core-side branches is
immaterial.  The total core-side share must remain $\gamma^\sigma$, but its
division among those branches may be arbitrary.

To see this, extend an IDEAL core mining action so that it may allocate hash
power among finitely many feasible targets.  In the simulation of
Lemma~\ref{lem:real-local-dominance}, part~2, the IDEAL core uses the same
relative allocation over the corresponding targets as the Frontier
core-side share.  Conditional on the successful target, the original
branchwise coupling then applies unchanged.

This extension requires no change to the certificate LP.  If an allocation
uses pure target actions $u_1,\ldots,u_K$ with weights
$\theta_k\geq0$ and $\sum_k\theta_k=1$, its exact centered action value is
\[
  \sum_{k=1}^K \theta_k\mathcal B_V(x,u_k).
\]
The Bellman inequality for every $u_k$ therefore implies the same inequality
for their convex combination.

\subsection{Stale Blocks}

Following Gervais et al.~\cite{gervais2016security}, suppose that a fraction
$s\in[0,1)$ of the blocks generated by every miner other than the deviating
miner becomes stale.  For a deviating miner of hash share $\alpha$, deleting
these stale discoveries leaves the original model with effective hash share
\[
  p_s(\alpha)
  =
  \frac{\alpha}{\alpha+(1-s)(1-\alpha)}.
\]
Equivalently, the nondeviating miners' hash power is multiplied by $1-s$ and
the remaining hash power is normalized.

The resulting threshold $\alpha_s^*$ satisfies
$p_s(\alpha_s^*)=\alpha^*$.  Solving gives
\[
  \alpha_s^*(\gamma^-,\gamma^+)
  =
  \frac{(1-s)\alpha^*(\gamma^-,\gamma^+)}
       {1-s\alpha^*(\gamma^-,\gamma^+)}.
\]
Thus the stale-block model requires no new strategic analysis: the same
transformation gives its threshold and maps the original lower and upper
bounds to stale-block lower and upper bounds.

%% file: sections/related-work.tex
\section{Related Work}
\label{sec:related-work}

The work most directly related to ours studies strategic mining and the
incentive-compatibility threshold. Eyal and Sirer
introduced selfish mining and showed that it can be profitable below the
majority threshold~\cite{eyal2014majority}. Nayak et al. enlarged this strategy
space with stubborn-mining policies~\cite{nayak2016stubborn}. Sapirshtein
et al. proposed a method for deriving optimal selfish-mining strategies using
a Markov Decision Process and computed near-tight numerical threshold
bounds~\cite{sapirshtein2016optimal}. Bar-Zur et al. improved the efficiency of
this MDP computation~\cite{barzur2020efficient}. Kiayias et al. derived
theoretical sufficient and necessary conditions for honest
mining~\cite{kiayias2016blockchain}. We continue this line of work while
providing theoretical lower and upper bounds in the broader REAL model.
Although our model is an $n$-miner game, the proof concerns unilateral
deviations from the all-Frontier profile; multiple simultaneously strategic
miners are studied in~\cite{bai2023multiple,marmolejo2019competing}.

Several extensions study related thresholds under different network or
protocol assumptions. Gervais et al. incorporate naturally stale blocks and
network parameters into proof-of-work security analysis~\cite{gervais2016security}.
Bahrani et al. study selfish mining under general stochastic
rewards~\cite{bahrani2025stochastic}.
Chatterjee et al. provide automated selfish-mining analysis for efficient
proof-system blockchains~\cite{chatterjee2024fully}. Niu and Feng, and Ritz and
Zugenmaier, study how Ethereum's uncle rewards affect selfish
mining~\cite{niu2019ethereum,ritz2018uncle}. Bitcoin-NG separates leader
election from transaction serialization~\cite{eyal2016bitcoinng}, and Niu
et al. study its protocol-specific incentive
compatibility~\cite{niu2020bitcoinng}.

Other important work addresses adjacent incentive and security questions.
Carlsten et al. show that fee-dominated rewards can destabilize Bitcoin after
block subsidies decline~\cite{carlsten2016instability}. Badertscher et al. use
rational protocol design to show that rational incentives can replace the
honest-majority assumption under suitable reward
conditions~\cite{badertscher2018rational}. Garay et al. establish the Bitcoin
backbone properties~\cite{garay2015backbone}. Ga\v{z}i et al. and Dembo et al.
independently derive tight bounds for longest-chain
consistency~\cite{gazi2020tight,dembo2020race}. These security thresholds bound
the adversarial power under which safety or liveness is preserved, whereas our
incentive-compatibility threshold asks whether Frontier mining is a best
response for miners. Propagation incentives are studied for transactions by
Babaioff et al.~\cite{babaioff2012red} and for blocks by Maeda
et al.~\cite{maeda2026propagation}. Schrijvers et al. separately characterize
the incentive compatibility of mining-pool reward
functions~\cite{schrijvers2016incentive}.

%% file: sections/conclusion.tex
\section{Conclusion}
\label{sec:conclusion}

We studied the incentive compatibility of the REAL model,
which allows a broad miner action space and asymmetric tie-breaking. Our
IDEAL-to-REAL reduction turns the certificate LP into a theoretical
lower-bound algorithm, while explicit deviation policies give an independent
theoretical upper-bound algorithm. Coordinatewise monotonicity and exact
rational evaluation extend these bounds to every
$(\gamma^-,\gamma^+)\in[0,1]^2$, with interval width at most
$9.98006\times10^{-4}$. Thus, the two algorithms give a near-tight theoretical
characterization of the incentive-compatibility threshold in the REAL model.